\documentclass[a4paper,UKenglish,cleveref,thm-restate]{lipics-v2021}

\RequirePackage{currfile}
\providecommand{\ifarxiv}[2]{#2}
\ifcurrfilename{arxiv-sort.tex}{%
  \renewcommand{\ifarxiv}[2]{#1}%
}{}

\ifarxiv{%
\pdfoutput=1
\hideLIPIcs
\nolinenumbers
\usepackage[appendix=inline,bibliography=common]{apxproof}
}{%
\usepackage[appendix=strip,bibliography=common]{apxproof}
}

\title{Symmetries in Sorting}

\author%
{Vikraman Choudhury}%
{University of Strathclyde, Glasgow, United Kingdom \and Universit\`{a} di Bologna, Bologna, Italy}%
{vikraman.choudhury@strath.ac.uk}%
{https://orcid.org/0000-0003-2030-8056}%
{Supported by the
 \href{https://cordis.europa.eu/programme/id/HORIZON.1.2/en}{EU Marie-Sk\l{}odowska-Curie}
 action ``ReGraDe-CS'', grant~\textnumero~\href{https://doi.org/10.3030/101106046}{101106046}.}
\author%
{Wind Wong}%
{Vrije Universiteit Amsterdam, Netherlands}%
{t.f.w.wong@vu.nl}%
{https://orcid.org/0009-0005-8306-7968}
{Supported by the European Research Council Starting Grant for the project
``SecuStack'', funded under the European Union's
\href{https://cordis.europa.eu/programme/id/HORIZON}{Horizon Europe}
grant~\textnumero~\href{https://cordis.europa.eu/project/id/101115046}{101115046}.}

\authorrunning{V. Choudhury and W. Wong}

\Copyright{Vikraman Choudhury and Wind Wong}



\category{} 

\ifarxiv{}{%
\relatedversion{} 
\relatedversiondetails[%
  linktext={arXiv:2512.07349},
  cite=choudhurySymmetriesSorting2025,
  ]{Extended version}{https://arxiv.org/abs/2512.07349}%
}

\supplement{}
\supplementdetails[%
  linktext={windtf/agda-symmetries},
  cite=choudhuryAgdasymmetries2025,
]{Formalisation}{https://github.com/windtf/agda-symmetries}%

\acknowledgements{
  This project was carried out as part of Wind Wong's undergraduate thesis at University of Glasgow,
  supervised by Simon Gay.
  We are grateful to Simon for his guidance and support.
  We are also grateful to Fredrik Nordvall Forsberg,
  and the reviewers of the HoTT/UF 2024 meeting,
  for their helpful comments and suggestions.
}

\EventEditors{Fredrik Nordvall Forsberg and James McKinna}
\EventNoEds{2}
\EventLongTitle{31st International Conference on Types for Proofs and Programs (TYPES 2025)}
\EventShortTitle{TYPES 2025}
\EventAcronym{TYPES 2025}
\EventYear{2025}
\EventDate{June 9--13, 2025}
\EventLocation{Glasgow, Scotland}
\EventLogo{}
\SeriesVolume{384}
\ArticleNo{0} 

\usepackage{subcaption}
\usepackage{float}
\usepackage{afterpage}

\usepackage{tikz}
\usepackage{xifthen}

\usepackage{amsmath,amsfonts,amsthm,amssymb}
\usepackage{newtxmath}
\usepackage{bbold}
\usepackage{mathpartir}
\usepackage{stmaryrd}
\usepackage{quiver}
\usepackage{ebproof}
\usepackage{listings}

\usepackage[
  style=english,
  english=american,
  french=guillemets,
  autopunct=true,
  csdisplay=true
]{csquotes}

\SetBlockEnvironment{innerquote}

\usepackage{xurl}

\usepackage{microtype}
\microtypesetup{
  activate={true,nocompatibility},
  final,
  tracking=true,
  kerning=all,
  spacing=true,
  protrusion=all,
  factor=1100,
  stretch=10,
  shrink=10
}

\usepackage{hott}
\usepackage{math}
\usepackage{code}
\usepackage{macros}

\begin{document}

\maketitle

\begin{abstract}
%
Sorting algorithms are fundamental to computer science,
and their correctness criteria are well understood as
rearranging elements of a list according to a specified total order on the underlying set of elements.
As mathematical functions, they are functions on lists that perform combinatorial operations on the representation of the input list.
In this paper, we study sorting algorithms conceptually as abstract sorting functions.

There is a canonical surjection from the free monoid on a set (lists of elements)
to the free commutative monoid on the same set (multisets of elements).
We show that sorting functions determine a section (right inverse) to this surjection satisfying two axioms,
that do not presuppose a total order on the underlying set.
Then, we establish an equivalence between (decidable) total orders on the underlying set and correct sorting functions.

The first part of the paper develops concepts from universal algebra from the point of view of functorial signatures,
and gives constructions of free monoids and free commutative monoids in (univalent) type theory.
Using these constructions, the second part of the paper develops the axiomatisation of sorting functions.
The paper uses informal mathematical language, and comes with an accompanying formalisation in Cubical Agda.

\end{abstract}


\section{Introduction}
\label{sec:introduction}

Consider a puzzle about sorting,
inspired by Dijkstra's Dutch National Flag problem~\cite[Ch.14]{dijkstraDisciplineProgramming1997}.
Suppose there are balls of three colours,
corresponding to the colours of the Dutch flag: red, white, and blue.
Given an unordered list (bag) of such balls, how many ways can you sort them into the Dutch flag?
\[
  \bag{
    \tikz[anchor=base, baseline]{%
      \foreach \x/\color in {1/red,2/red,3/blue,4/white,5/blue,6/red,7/white,8/blue} {
          \node[circle,draw,fill=\color,line width=1pt] at (\x,0) {\phantom{\tiny\x}};
          \ifthenelse{\NOT 1 = \x}{\node at ({\x-0.5},0) {,};}{}
        }
    }
  }
\]
Obviously there is only one way, decided by the order the colours appear in the Dutch flag:
$\term{red} < \term{white} < \term{blue}$.
\[
  [
    \tikz[anchor=base, baseline]{%
      \foreach \x/\color in {1/red,2/red,3/red,4/white,5/white,6/blue,7/blue,8/blue} {
          \node[circle,draw,fill=\color,line width=1pt] at (\x,0) {\phantom{\tiny\x}};
          \ifthenelse{\NOT 1 = \x}{\node at ({\x-0.5},0) {,};}{}
        }
    }
  ]
\]
What if we are avid enjoyers of vexillology who also want to consider other flags?
How many ways can we sort our unordered list?
To answer this, we distinguish between a sorting \emph{algorithm} (a computational procedure) and a sorting \emph{function} (a mathematical mapping). 
Since we sort the elements of a carrier set $A$ directly (so stability and metadata are not concerns), a sorting function must produce a unique, deterministic output list for each input bag. 
Thus, ``how many ways to sort'' means how many extensionally distinct sorting functions can be defined. 
Because there are $3! = 6$ permutations of $\{\term{red}, \term{white}, \term{blue}\}$, there are 6 possible orderings, corresponding to 6 categories of tricolour flags
(see~\href{https://en.wikipedia.org/wiki/List_of_flags_with_blue,_red,_and_white_stripes#Triband}{Wikipedia}):
\begin{center}
  \foreach \colorA/\colorB/\colorC in {red/white/blue, red/blue/white, white/red/blue, white/blue/red, blue/red/white, blue/white/red}{
      \begin{tikzpicture}[scale=0.5]
        \begin{flagdescription}{3/4}
          \hstripesIII{\colorA}{\colorB}{\colorC}
          \framecode{}
        \end{flagdescription}
      \end{tikzpicture}
    }
\end{center}
For a carrier set with 3 elements, the 6 possible total orderings yield exactly 6 \textit{extensionally correct} sorting functions. 
Intuitively, a \emph{correct} sorting function takes a bag and returns a unique, sorted list of the same elements. 
Extensional correctness therefore requires the output list to be a permutation of the input elements, and sorted according to a fixed total order. 
Formally, there is a bijection between total orderings on a carrier set $A$ and correct sorting functions on lists over $A$. 
Our main contribution is showing how a sorting function can be correctly axiomatised and formalised from this point of view.

\paragraph*{Outline and Contributions}

The paper is organised as follows:
\begin{myitemize}
  \item In~\cref{sec:universal-algebra}, we describe a formalisation of universal algebra developed from the point of view
        of functorial signatures, the definition and universal property of free algebras,
        and equational systems over signatures.
        These are elementary constructions but necessary groundwork for us,
        and we give an informal exposition in type theory.
  \item In~\cref{sec:monoids}, we give constructions of free monoids, and proofs of their universal property.
        Following this, in~\cref{sec:commutative-monoids}, we add commutativity to each representation of free monoids,
        and extend the proofs of universal property from free monoids to free commutative monoids.
        These constructions are standard, and we formalise by using our general framework.
  \item In \cref{sec:application}, we build on the previous constructions and study sorting functions.
        The main result connects total orders, sorting, and commutativity,
        by proving an equivalence between decidable total orders on a carrier set $A$,
        and correct sorting functions defined using free monoids and free commutative monoids over $A$.
  \item Finally, \cref{sec:discussion} discusses aspects of the formalisation, related and future work.
\end{myitemize}
These sections are largely independent. Readers interested in universal algebra can start at~\cref{sec:universal-algebra}; those focused on free (commutative) monoids can skip to~\cref{sec:monoids,sec:commutative-monoids}; and those interested primarily in sorting can proceed directly to~\cref{sec:application}.
All results are formalised in Cubical Agda, available online at~\cite{choudhuryAgdasymmetries2025}.
We use informal mathematical language in the paper to keep it accessible.
\ifarxiv{}{An extended version with additional details is on arXiv~\cite{choudhurySymmetriesSorting2025}.}

No prior exposure to Homotopy Type Theory (HoTT) or Univalence is required. 
We assume only basic familiarity with dependent type theory and functional programming (such as Agda, Coq, or Haskell). 
While the formalisation is written in Cubical Agda and uses higher inductive types and univalence for quotient constructions, the paper's text is presented in informal mathematical language.

\section{Universal Algebra}
\label{sec:universal-algebra}

We develop some basic notions from universal algebra and equational
logic~\cite{birkhoffStructureAbstractAlgebras1935},
which gives us a vocabulary and framework to express our results in.
The point of view we take is the standard category-theoretic approach to universal algebra~\cite{waltersCategoricalApproachUniversal1970},
which is an alternative perspective to the Lawvere theory or abstract clone point of view~\cite{lawvereFunctorialSemanticsAlgebraic1963,hylandCategoryTheoreticUnderstanding2007}.
We keep a running example of monoids in mind, while explaining and defining the abstract concepts.

\subsection{Algebras}
\label{sec:universal-algebra:algebras}

\begin{definition}[\alink{definition}{Signature}{signature}]
    \label{algebra:signature}
    \label{def:signature}
    A signature, denoted $\sig$, is a (dependent) pair consisting of
    a set of operations, $\op\colon \Set$, and
    an arity function for each operation, $\ar\colon \op \to \Set$.
\end{definition}

\begin{example}
    A premonoid is a set with an identity element (a nullary operation), and a binary multiplication operation,
    with signature $\sigMon \defeq (\Fin[2],\lambda \{0 \mapsto \Fin[0] ; 1 \mapsto \Fin[2] \})$.
    Informally, we denote the two operations as a tuple $(e,\mult)$.
\end{example}

\noindent
Every signature $\sig$ induces a signature functor $\Sig$ on $\Set$.
\begin{definition}[\alink{definition}{Signature functor}{signature-functor}]
    \label{def:signature-functor}
    \[
        X \mapsto \dsum{o:\op}{X^{\ar(o)}}
        \qquad
        X \xto{f} Y \mapsto
        \dsum{o:\op}{X^{\ar(o)}}
        \xto{(o, \blank \comp f)}
        \dsum{o:\op}{Y^{\ar(o)}}
    \]
\end{definition}

\begin{example}
    The signature functor for monoids, $\SigMon$, assigns to a carrier set $X$,
    the set of inputs for each operation.
    Expanding the dependent product on $\Fin[2]$, we obtain a coproduct of sets:
    $\SigMon(X) \eqv X^{\Fin[0]} + X^{\Fin[2]} \eqv \unitt + X \times X$.
\end{example}
A $\sig$-structure is given by a carrier set, with functions interpreting each operation symbol.
The signature functor applied to a carrier set gives the inputs to each operation,
and the output is simply a map back to the carrier set.
These two pieces of data constitute an algebra for the $\Sig$ functor.
We write $\str{X}$ for a $\sig$-structure with carrier set $X$, following the model-theoretic notational convention.

\begin{definition}[\alink{definition}{Structure}{structure}]
    \label{algebra:struct}
    A $\sig$-structure $\str{X}$ is an $\Sig$-algebra, that is, a pair consisting of
    a carrier set $X$, and
    an algebra map $\alpha_{X}\colon \Sig(X) \to X$.
\end{definition}

\begin{example}
    Concretely, an $\SigMon$-algebra has the type
    \[
        \alpha_{X} : \SigMon(X) \to X
        \quad\eqv\quad
        (\unitt + X \times X) \to X
        \quad\eqv\quad
        (\unitt \to X) \times (X \times X \to X)
    \]
    which is the pair of functions interpreting the two operations.
    Natural numbers $\Nat$ with $(0, +)$ or $(1, \times)$ are examples of monoid structures.
\end{example}

\begin{definition}[\alink{definition}{$\sig$-Homomorphism}{sig-homomorphism}]
    A homomorphism between $\sig$-structures $\str{X}$ and $\str{Y}$ is a morphism of $\Sig$-algebras,
    that is, a map $f : X \to Y$ making the following diagram commute:
    \[
        \begin{tikzcd}
            \Sig(X) \arrow[r, "\alpha_{X}"] \arrow[d, "\Sig(f)"']
            & X \arrow[d, "f"] \\
            \Sig(Y) \arrow[r, "\alpha_{Y}"']
            & Y
        \end{tikzcd}
    \]
\end{definition}

\begin{example}
    Given two monoid structures \(\str{X}\) and \(\str{Y}\), the top half denotes:
    \({\unitt + (X \times X) \xto{\alpha_{X}} X \xto{f} Y}\), which applies \(f\) to the output of each operation,
    and the bottom half denotes the map:
    \({\unitt + (X \times X) \xto{\SigMon(f)} \unitt + (Y \times Y) \xto{\alpha_{Y}} Y}\).
    %
    In other words, a homomorphism between \(\str{X}\) and \(\str{Y}\) is a map $f$ on the carrier sets that commutes with the
    interpretation of the monoid operations, or simply, preserves the monoid structure.
\end{example}
For a fixed signature $\sig$,
the category of $\Sig$-algebras and their morphisms form a category of algebras,
written $\SigAlg$, or simply, $\sigAlg$,
given by the obvious definitions of identity and composition for the underlying functions.

\subsection{Free Algebras}
\label{sec:universal-algebra:free-algebras}

The concrete category $\sigAlg$ of structured sets and structure-preserving maps
admits a forgetful functor $U : \sigAlg \to \Set$.
In our notation, $U(\str{X})$ is simply $X$, a fact we exploit to simplify our notation and formalisation.
The left adjoint to this forgetful functor is the free algebra construction,
also known as the term algebra or the \emph{absolutely free} algebra without equations~\cite{reiterman1991free}.
We rephrase this in more concrete terms.

\begin{definition}[\alink{definition}{Free Algebras}{free-algebras}]
    \label{def:free-algebras}
    A free $\sig$-algebra construction consists of the following:
    \begin{itemize}
        \item a set $F(X)$, for every set $X$,
        \item a $\sig$-structure on $F(X)$, written as $\str{F}(X)$,
        \item a universal map $\eta_X : X \to F(X)$, for every $X$, such that,
        \item for any $\sig$-algebra $\str{Y}$, the operation
              assigning to each homomorphism $f : \str{F}(X) \to \str{Y}$,
              the map ${f \comp \eta_X : X \to Y}$ (or, pre-composition with $\eta_X$),
              is an equivalence.
    \end{itemize}
\end{definition}
In other words,
we ask for a bijection between the set of homomorphisms out of the free algebra to any other algebra,
and the set of functions from the carrier set of the free algebra to the carrier set of the other algebra.
There should be no more data in homomorphisms out of the free algebra than there is in functions out of
the carrier set, which is the property of \emph{freeness}.
The inverse operation to post-composing with $\eta_X$ extends a function to a homomorphism.
\begin{definition}[\alink{definition}{Universal extension}{universal-extension}]
    \label{def:universal-extension}
    The universal extension of a function $f : X \to Y$ to a homomorphism out of the free $\sig$-algebra on $X$ is written
    as $\ext{f} : \str{F}(X) \to \str{Y}$.
    It satisfies the identities: $\ext{f} \comp \eta_X = f$, $\ext{\eta_{X}} = \idfunc_{\str{F}(X)}$,
    and $\ext{(\ext{g} \comp f)} = \ext{g} \comp \ext{f}$.
    The universal morphism $\eta_X$ is highlighted in \redtext{red}.
    \[\begin{tikzcd}[ampersand replacement=\&]
            X \\
            \\
            Y
            \arrow["f"', from=1-1, to=3-1]
        \end{tikzcd}
        \quad\mapsto\quad
        \begin{tikzcd}[ampersand replacement=\&]
            X \&\& {F(X)} \\
            \\
            \&\& {Y}
            \arrow["\eta_X", color={solarized-red}, from=1-1, to=1-3]
            \arrow["f"', from=1-1, to=3-3]
            \arrow["\ext{f}", dotted, from=1-3, to=3-3]
        \end{tikzcd}\]
\end{definition}
Free algebra constructions are canonically equivalent.
\begin{propositionrep}[\alink{proposition}{}{free-algebras-unique}]
    \label{lem:free-algebras-unique}
    If $\str{F}(X)$ and $\str{G}(X)$ are free $\sig$-algebras on $X$,
    then $\str{F}(X) \eqv \str{G}(X)$ in $\sigAlg$, naturally in $X$, via the canonical extensions.
\end{propositionrep}
\begin{proof}
    By extending $\eta_X$ for each free construction,
    we have maps in each direction:
    ${\ext{G\fdot\eta_{X}} : \str{F}(X) \to \str{G}(X)}$, and vice versa.
    Finally, using~\cref{def:universal-extension}, we have
    \(
    {\ext{F\fdot\eta_{X}} \comp \ext{G\fdot\eta_{X}}} =
    {\ext{(\ext{F\fdot\eta_{X}} \comp G\fdot\eta_{X})}} =
    {\ext{F\fdot\eta_{X}}} =
    {\idfunc_{\str{F}(X)}}
    \).
\end{proof}
\begin{toappendix}
    \begin{proposition}[\alink{proposition}{}{free-algebra-monad}]
        \label{prop:free-algebra-monad}
        The free algebra construction produces a monad on $\Set$.
    \end{proposition}
    \begin{proof}
        We have unit $\eta$ and extension $\ext{(\blank)}$ from the free algebra construction.
        $F$ becomes an endofunctor on $\Set$ where the action on functions is given by
        $f : X \to Y \mapsto \ext{(\eta_{Y} \comp f)} : F(X) \to F(Y)$.
        The monad unit is given by $\eta$,
        and multiplication given by $\mu_{X} \defeq \ext{\idfunc_{F(X)}} : F(F(X)) \to F(X)$.
        The monad laws follow from the identities of universal extension.
    \end{proof}

    \begin{proposition}[\alink{proposition}{}{free-algebra-colimits}]
        \label{prop:free-algebra-colimits}
        The following properties of free algebras on $\emptyt$, $\unitt$, and $+$ hold:
        \begin{itemize}
            \item if $\sig$ has one constant symbol, then $\str{F}(\emptyt)$ is contractible,
            \item the type of algebra structures on $\unitt$ is contractible,
            \item $\str{F}(X + Y)$ is the coproduct of $\str{F}(X)$ and $\str{F}(Y)$ in $\sigAlg$, that is:
                  \[
                      \sigAlg(\str{F}(X + Y), \str{Z}) \eqv \sigAlg(\str{F}(X), \str{Z}) \times \sigAlg(\str{F}(Y), \str{Z}) \enspace.
                  \]
        \end{itemize}
    \end{proposition}
    \begin{proof}
        $F$ being a left adjoint, preserves coproducts.
        This makes $\str{F}(\emptyt)$ initial in $\sigAlg$.
        Note that~\(\str{F}(\emptyt)\) is inhabited by all the constant symbols in the signature,
        so if there is one constant symbol, it becomes contractible.
        $\str{F}(\unitt) \to \unitt$ is contractible because $\unitt$ is terminal in $\Set$.
    \end{proof}
\end{toappendix}
We have discussed specifications of free algebras, but not actually given a construction.
In type theory, \emph{free} constructions are often given by inductive types,
where the constructors are the pieces of data that freely generate the structure,
and the type-theoretic induction principle enforces the category-theoretic universal property.

\begin{definition}[\alink{definition}{Construction of Free $\sig$-algebras}{free-algebra-construction}]
    \label{algebra:tree}
    \label{def:free-algebra-construction}
    The free $\sig$-algebra on a type $X$ is given by the inductive type $\type{Tree}(X)$,
    generated by two constructors:
    \begin{align*}
        \term{leaf} & \colon X \to \type{Tree}(X)                      \\
        \term{node} & \colon F_\sig(\type{Tree}(X)) \to \type{Tree}(X)
    \end{align*}
\end{definition}
The constructors \inline{leaf} and \inline{node} are
the generators for the universal map, and the algebra map, respectively.
Together, they describe the type of abstract syntax trees for terms in the signature $\sig$
-- the leaves are the free variables, and the nodes are the branching operations of the tree,
marked by the operations in $\sig$.
\begin{toappendix}
    \begin{example}
        Trees for $\sigMon$ look like:
        \begin{center}
            \scalebox{0.9}
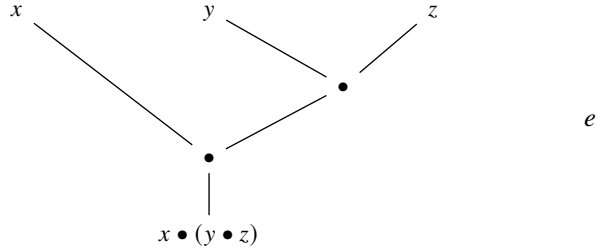
{%
                \begin{tikzcd}[ampersand replacement=\&]
                    x \&\& y \&\& z \\
                    \&\&\& \mult \\
                    \&\& \mult \\
                    \&\& {x \mult (y \mult z)}
                    \arrow[no head, from=1-1, to=3-3]
                    \arrow[no head, from=1-3, to=2-4]
                    \arrow[no head, from=1-5, to=2-4]
                    \arrow[no head, from=2-4, to=3-3]
                    \arrow[no head, from=3-3, to=4-3]
                \end{tikzcd}
            }
            \qquad\qquad
            \scalebox{1}{%
                \begin{tikzcd}[ampersand replacement=\&]
                    e
                \end{tikzcd}
            }
        \end{center}
    \end{example}
\end{toappendix}
\begin{proposition}[\alink{proposition}{}{free-algebra-construction-is}]
    \label{prop:free-algebra-construction-is}
    $(\type{Tree}(X), \term{leaf})$ is the free $\sig$-algebra on~$X$. 
\end{proposition}

\subsection{Equations}
\label{sec:universal-algebra:equations}

The algebraic framework described so far only captures operations, not equations.
These algebras are \emph{lawless} (or \emph{wild} or \emph{absolutely free}) --
$\SigMon$-algebras are premonoids rather than monoids,
and $\str{F}_{\sigMon}$-algebras are free premonoids, not free monoids,
since they are missing the unit and associativity laws.
\begin{toappendix}
    For example, by associativity, these two trees of $(\Nat, 0, +)$ should be identified as equal.
    \begin{center}
        \scalebox{0.7}
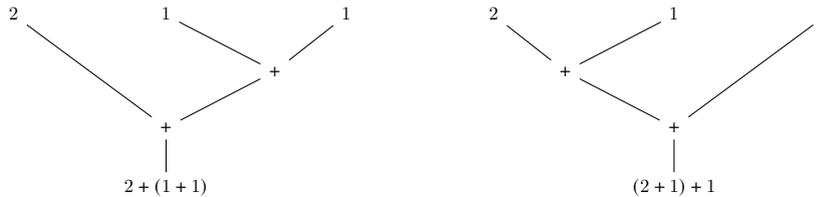
{%
            \begin{tikzcd}[ampersand replacement=\&]
                2 \&\& 1 \&\& 1 \\
                \&\&\& {+} \\
                \&\& {+} \\
                \&\& 2 + (1 + 1)
                \arrow[no head, from=1-1, to=3-3]
                \arrow[no head, from=1-3, to=2-4]
                \arrow[no head, from=1-5, to=2-4]
                \arrow[no head, from=2-4, to=3-3]
                \arrow[no head, from=3-3, to=4-3]
            \end{tikzcd}
        }
        \qquad\qquad
        \scalebox{0.7}{%
            \begin{tikzcd}[ampersand replacement=\&]
                2 \&\& 1 \&\& 1 \\
                \& {+} \\
                \&\& {+} \\
                \&\& {(2 + 1) + 1}
                \arrow[no head, from=3-3, to=4-3]
                \arrow[no head, from=1-5, to=3-3]
                \arrow[no head, from=1-1, to=2-2]
                \arrow[no head, from=1-3, to=2-2]
                \arrow[no head, from=2-2, to=3-3]
            \end{tikzcd}
        }
    \end{center}
\end{toappendix}
To impose equations on the generated abstract syntax trees, we adopt the point of view of equational logic.
\begin{definition}[\alink{definition}{Term Arity}{term-arity}]
    \label{def:term-arity}
    A \emph{term arity}, denoted $\eqsig$, is a (dependent) pair consisting of:
    a set of names for equations, $\eqop : \Set$, and
    an arity of free variables for each equation, $\eqfv : \eqop \to \Set$.
\end{definition}
This will be used to define a \emph{system of equations} (\cref{def:system-of-equations}).
\begin{example}
    The term arity for monoids $\eqsigMon$ is:
    $(\Fin[3],\lambda \{0 \mapsto \Fin[1] ; 1 \mapsto \Fin[1] ; 2 \mapsto \Fin[3] \})$.
    The three equations are the left and right unit laws, and the associativity law --
    a 3-element set of names $\{ \term{unitl}, \term{unitr}, \term{assoc} \}$.
    The two unit laws use one free variable,
    and the associativity law uses three free variables.
\end{example}
Similar to the term signature functor in~\cref{def:signature-functor}, this produces a \emph{term arity functor} on $\Set$.
\begin{definition}[\alink{definition}{Term Arity Functor}{term-arity-functor}]
    \label{def:term-arity-functor}
    \begin{gather*}
        X \mapsto \dsum{e:\eqop}{X^{\eqfv(e)}}
        \qquad
        X \xto{f} Y \mapsto
        \dsum{e:\eqop}{X^{\eqfv(e)}}
        \xto{(e, \blank \comp f)}
        \dsum{e:\eqop}{Y^{\eqfv(e)}}
    \end{gather*}
\end{definition}
To build equations out of this,
we use the $\sig$-operations and construct trees for the left and right-hand sides of each equation using the
free variables available.
\begin{definition}[\alink{definition}{System of Equations}{system-of-equations}]
    \label{def:system-of-equations}
    A system of equations over a signature $(\sig,\eqsig)$ is a pair of natural transformations:
    \[
        \eqleft, \eqright \colon \EqSig \natto \str{F}_{\sig} \enspace.
    \]
    For any set (of variables) $V$,
    this gives a pair of functions $\eqleft_{V}, \eqright_{V}\colon \EqSig(V) \to \str{F}_{\sig}(V)$,
    and naturality ensures correctness of renaming.
\end{definition}
\begin{toappendix}
    \begin{example}
        Given $x\colon V$, $\eqleft_{V}(\term{unitl}, (x)),\,\eqright_{V}(\term{unitl}, (x))$
        are defined as:
        
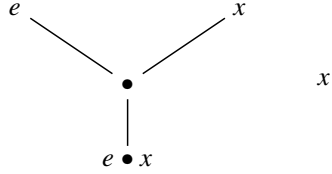
\begin{center}
            \begin{tikzcd}[ampersand replacement=\&,cramped]
                e \&\& x \\
                \& \mult \\
                \& {e \mult x}
                \arrow[no head, from=1-1, to=2-2]
                \arrow[no head, from=1-3, to=2-2]
                \arrow[no head, from=2-2, to=3-2]
            \end{tikzcd}
            \hspace{2em}
            \begin{tikzcd}[ampersand replacement=\&,cramped]
                x
            \end{tikzcd}
        \end{center}
        Given $x\colon V$, $\eqleft_{V}(\term{unitr}, (x)),\,\eqright_{V}(\term{unitr}, (x))$
        are defined as:
        
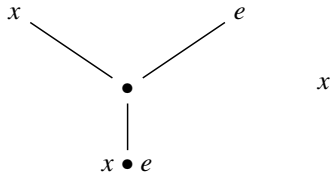
\begin{center}
            \begin{tikzcd}[ampersand replacement=\&,cramped]
                x \&\& e \\
                \& \mult \\
                \& {x \mult e}
                \arrow[no head, from=1-1, to=2-2]
                \arrow[no head, from=1-3, to=2-2]
                \arrow[no head, from=2-2, to=3-2]
            \end{tikzcd}
            \hspace{2em}
            \begin{tikzcd}[ampersand replacement=\&,cramped]
                x
            \end{tikzcd}
        \end{center}
        Given $x, y, z : V$, $\eqleft_{V}(\term{assocr}, (x,y,z)),\,\eqright_{V}(\term{assocr}, (x,y,z))$
        are defined as:
        \begin{center}
            \scalebox{0.7}
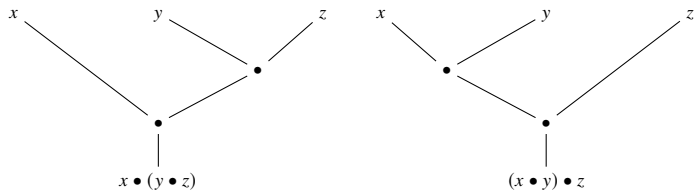
{
                \begin{tikzcd}[ampersand replacement=\&,cramped]
                    x \&\& y \&\& z \\
                    \&\&\& {\mult} \\
                    \&\& {\mult} \\
                    \&\& {x \mult (y \mult z)}
                    \arrow[no head, from=3-3, to=4-3]
                    \arrow[no head, from=2-4, to=3-3]
                    \arrow[no head, from=1-3, to=2-4]
                    \arrow[no head, from=1-5, to=2-4]
                    \arrow[no head, from=1-1, to=3-3]
                \end{tikzcd}
                \hspace{2em}
                \begin{tikzcd}[ampersand replacement=\&,cramped]
                    x \&\& y \&\& z \\
                    \& {\mult} \\
                    \&\& {\mult} \\
                    \&\& {(x \mult y) \mult z}
                    \arrow[no head, from=3-3, to=4-3]
                    \arrow[no head, from=1-5, to=3-3]
                    \arrow[no head, from=1-1, to=2-2]
                    \arrow[no head, from=1-3, to=2-2]
                    \arrow[no head, from=2-2, to=3-3]
                \end{tikzcd}
            }
        \end{center}
    \end{example}
\end{toappendix}
Finally, we say how a given $\sig$-structure $\str{X}$
\emph{satisfies} the system of equations $T_{(\sig,\eqsig)}$.
We need to assign a value to each free variable in the equation
in the carrier set, using a valuation function $\rho : V \to X$.
Given such an assignment, we evaluate the left and right trees of the equation,
by extending $\rho\,$ with~\cref{def:universal-extension},
that is by construction, a homomorphism from $\str{F}(V)$ to $\str{X}$.
To satisfy an equation, these two evaluations should agree.
\begin{definition}[\alink{definition}{Satisfaction $\str{X} \entails T$}{satisfaction}]
    A $\sig$-structure $\str{X}$ satisfies the system of equations $T_{(\sig,\eqsig)}$ if for every set $V$,
    and every assignment $\rho : V \to X$, $\ext{\rho}$ is a (co)fork:
    \[\begin{tikzcd}[ampersand replacement=\&,cramped]
            {\EqSig(V)} \&\& {\str{F}(V)} \&\& {\str{X}}
            \arrow["{\ext{\rho}}", from=1-3, to=1-5]
            \arrow["{\eqleft_{V}}", shift left=3, from=1-1, to=1-3]
            \arrow["{\eqright_{V}}"', shift right=3, from=1-1, to=1-3]
        \end{tikzcd}\]
\end{definition}
Given a signature $(\sig,\eqsig)$ with a system of equations $T_{(\sig,\eqsig)}$,
the $\sig$-algebras satisfying $T_{(\sig,\eqsig)}$, or varieties, form a full subcategory of $\sigAlg$.
While free algebras exist for any signature $\sig$ as in~\cref{def:free-algebra-construction},
constructions of free varieties for arbitrary systems of equations require non-constructive
principles~\cite[\S~7, pg.142]{blassWordsFreeAlgebras1983},
in particular, the arity sets need to be projective -- we do not pursue this matter further.
%
%
We have a rudimentary framework for universal algebra and equational logic,
which gives us enough tools to develop the next sections.

\section{Constructions of Free Monoids}
\label{sec:monoids}

In this section, we consider various constructions of free monoids in type theory,
with proofs of their universal property.
Since each construction satisfies the same categorical universal property,
by~\cref{lem:free-algebras-unique},
these are canonically equivalent (hence equal, by univalence) as types (and as monoids),
allowing us to transport proofs between them.
There are countably many such representations of free monoids, allowing us to
choose the most convenient one for a given task.
Using the universal property allows us to define and prove our programs correct in one go,
which we use in~\cref{sec:application}.

\subsection{Lists}
\label{sec:lists}

Cons-lists are simple inductive datatypes, well-known to functional programmers,
and are a common representation of free monoids in programming languages.
More generally, they correspond to
list objects~\cite{cockettListarithmeticDistributiveCategories1990,maiettiJoyalsArithmeticUniverse2010}
and algebraically-free monoids~\cite{kellyUnifiedTreatmentTransfinite1980} in category theory.
\begin{definition}[\alink{definition}{Lists}{lists}]
    \label{def:lists}
    $\type{List(A)}$ is generated by the following constructors:
    \begin{align*}
        \nil              & : \List(A)                    \\
        \blank\cons\blank & : A \to \List(A) \to \List(A)
    \end{align*}
\end{definition}
The (universal) generators map is the singleton: $\eta_A(a) \defeq [a] \jdgeq a \cons \nil$,
the identity element is the empty list~\inline{nil},
and the monoid multiplication $\mult$ is given by list concatenation.
\begin{toappendix}
    \begin{definition}[\alink{definition}{Concatenation}{list-concatenation}]
        We define the concatenation operation $\mult : \List(A) \to \List(A) \to \List(A)$,
        by recursion on the first argument:
        \begin{align*}
            \nil \mult \ys          & = \ys                     \\
            (x \cons \xs) \mult \ys & = x \cons (\xs \mult \ys)
        \end{align*}
    \end{definition}
    The proof that $\mult$ satisfies monoid laws is straightforward (see the formalisation for details).

    \begin{definition}[\alink{definition}{Universal extension}{list-universal-extension}]
        For any monoid $\str{X}$, and given a map $f : A \to X$,
        we define the extension $\ext{f} : \List(A) \to \mathfrak{X}$ by recursion on the list:
        \begin{align*}
            \ext{f}(\nil)        & = e                       \\
            \ext{f}(x \cons \xs) & = f(x) \mult \ext{f}(\xs)
        \end{align*}
    \end{definition}

    \begin{proposition}[\alink{proposition}{}{ext-lifts-homomorphism}]
        $\ext{(\blank)}$ lifts a function $f : A \to X$ to a monoid homomorphism $\ext{f} : \List(A) \to \mathfrak{X}$.
    \end{proposition}

    \begin{proof}
        To show that $\ext{f}$ is a monoid homomorphism,
        we need to show $\ext{f}(\xs \mult \ys) = \ext{f}(\xs) \mult \ext{f}(\ys)$.
        We do so by induction on $\xs$.

        Case $\nil$:
        \begin{align*}
            \ext{f}(\nil \mult \ys)
             & = \ext{f}(\ys)                     & \text{by definition of concatenation}    \\
             & = e \mult \ext{f}(\ys)             & \text{by unit law}                       \\
             & = \ext{f}(\nil) \mult \ext{f}(\ys) & \text{by definition of $\ext{(\blank)}$}
        \end{align*}

        Case $x \cons \xs$:
        \begin{align*}
            \ext{f}((x \cons \xs) \mult \ys)
             & = \ext{f}(([ x ] \mult \xs) \mult \ys)         & \text{by definition of concatenation}    \\
             & = \ext{f}([ x ] \mult (\xs \mult \ys))         & \text{by associativity}                  \\
             & = \ext{f}(x \cons (\xs \mult \ ys))            & \text{by definition of concatenation}    \\
             & = f(x) \mult \ext{f}(\xs \mult \ys)            & \text{by definition of $\ext{(\blank)}$} \\
             & = f(x) \mult (\ext{f}(\xs) \mult \ext{f}(\ys)) & \text{by induction}                      \\
             & = (f(x) \mult \ext{f}(\xs)) \mult \ext{f}(\ys) & \text{by associativity}                  \\
             & = \ext{f}(x \cons \xs) \mult \ext{f}(\ys)      & \text{by definition of $\ext{(\blank)}$} \\
        \end{align*}
        Therefore, $\ext{(\blank)}$ does correctly lift a function to a monoid homomorphism.
    \end{proof}
\end{toappendix}
\begin{propositionrep}[\alink{proposition}{Universal property for List}{universal-property-for-list}]
    $(\List(A),\eta_A)$ is the free monoid on $A$.
\end{propositionrep}

\begin{proof}
    To show that $\ext{(\blank)}$ is an inverse to $\blank \comp \eta_A$,
    we first show $\ext{(\blank)}$ is the right inverse to $\blank \comp \eta_A$.
    For all $f$ and $x$, $(\ext{f} \circ \eta_A)(x) = \ext{f}(x \cons \nil) = f(x) \mult e = f(x)$,
    therefore by function extensionality, for any $f$, $\ext{f} \circ \eta_A = f$,
    and $(\blank \circ \eta_A) \comp \ext{(\blank)} = id$.

    To show $\ext{(\blank)}$ is the left inverse to $\blank \comp \eta_A$, we need to prove
    for any monoid homomorphism $f : \List(A) \to \mathfrak{X}$, $\ext{(f \comp \eta_A)}(\xs) = f(\xs)$.
    We can do so by induction on $\xs$.

    Case $\nil$:
    \begin{align*}
        \ext{(f \comp \eta_A)}(\nil)
         & = e       & \text{by definition of $\ext{(\blank)}$} \\
         & = f(\nil) & \text{by homomorphism properties of $f$}
    \end{align*}

    Case $x \cons \xs$:
    \begin{align*}
        \ext{(f \comp \eta_A)}(x \cons \xs)
         & = (f \comp \eta_A)(x) \mult \ext{(f \comp \eta_A)}(\xs) & \text{by definition of $\ext{(\blank)}$} \\
         & = (f \comp \eta_A)(x) \mult f(\xs)                      & \text{by induction}                      \\
         & = f([x]) \mult f(\xs)                                   & \text{by definition of $\eta_A$}         \\
         & = f([x] \mult \xs)                                      & \text{by homomorphism properties of $f$} \\
         & = f(x \cons \xs)                                        & \text{by definition of concatenation}
    \end{align*}
    By function extensionality, $\ext{(\blank)} \comp (\blank \circ \eta_A) = id$.
    Therefore, $\ext{(\blank)}$ and $(\blank) \circ [\_]$ are inverse of each other.

    We have now shown that $(\blank) \comp \eta_A$ is an equivalence from
    monoid homomorphisms $\List(A) \to \mathfrak{X}$
    to set functions $A \to X$, and its inverse is given by $\ext{(\blank)}$, which maps set
    functions $A \to X$ to monoid homomorphisms $\List(A) \to \mathfrak{X}$. Therefore, $\List$ is indeed
    the free monoid.
\end{proof}

\subsection{Arrays}
\label{sec:arrays}

An alternate (non-inductive) representation of the free monoid on a carrier set,
or alphabet $A$, is $A^{\ast}$,
the set of all finite strings or sequences of characters \emph{drawn} from $A$,
as in~\cite{dubucFreeMonoids1974}.
We call this an \emph{array},
which is a pair of a natural number $n$, denoting the length of the array,
and a lookup (or index) function $\Fin[n] \to A$, mapping indices to elements of $A$.
In type theory, this is also understood as an extension of a container~\cite{abbottCategoriesContainers2003},
where the type of shapes is a $n\colon\Nat$, and the type of positions for each shape is $\Fin[n]$,
where the array elements are stored.
\begin{definition}[\alink{definition}{Arrays}{arrays}]
    \label{def:arrays}
    \[
        \Array(A) \defeq \dsum{n : \Nat}{(\Fin[n] \to A)}
    \]
\end{definition}
For example, $(3, \lambda\{ 0 \mapsto 3, 1 \mapsto 1, 2 \mapsto 2 \})$
represents the same list as $[3, 1, 2]$.
The (universal) generators map is the singleton: $\eta_A(a) = (1, \lambda\{ 0 \mapsto a \})$,
the identity element is $(0, \lambda\{\})$
and the monoid operation $\mult$ is given by array concatenation.
\begin{toappendix}
    \begin{lemma}[\alink{lemma}{}{array-zero-is-id}]
        \label{array:zero-is-id}
        Zero-length arrays $(0, f)$ are contractible.
    \end{lemma}
    \begin{proof}
        We need to show $f : \Fin[0] \to A$ is equal to $\lambda\{\}$.
        By function extensionality this amounts to showing for all $x : \emptyt$, $f(x) = (\lambda\{\})(x)$,
        which holds by absurdity elimination on $x$.
        Therefore, any array $(0, f)$ is equal to $(0, \lambda\{\})$.
    \end{proof}
\end{toappendix}
\begin{definition}[\alink{definition}{Concatenation}{concatenation}]
    The concatenation operation $\mult$, 
    is defined below, where $\oplus : (\Fin[n] \to A) \to (\Fin[m] \to A) \to (\Fin[n+m] \to\nolinebreak A)$
    combines the two lookup functions.
    \begin{gather*}
        (n , f) \mult (m , g) = (n + m , f \oplus g)
        \qquad
        (f \oplus g)(k)             = \begin{cases}
            f(k)     & \text{if}\ k < n \\
            g(k - n) & \text{otherwise}
        \end{cases}
    \end{gather*}
\end{definition}
\begin{toappendix}
    \begin{proposition}[\alink{proposition}{}{arrays-monoid}]
        $(\Array(A), \mult)$ is a monoid.
    \end{proposition}

    \begin{proof}
        To show $\Array$ satisfies left unit,
        we want to show $(0, \lambda\{\}) \mult (n, f) = (n, f)$.
        \begin{align*}
            (0 , \lambda\{\}) \mult (n , f) & = (0 + n , \lambda\{\} \oplus f)          \\
            (\lambda\{\} \oplus f)(k)       & = \begin{cases}
                                                    (\lambda\{\})(k) & \text{if}\ k < 0 \\
                                                    f(k - 0)         & \text{otherwise}
                                                \end{cases}
        \end{align*}
        It is trivial to see the length matches: $0 + n = n$. We also need to show $\lambda\{\} \oplus f = f$.
        Since $n < 0$ for any $n : \mathbb{N}$ is impossible, $(\lambda\{\} \oplus f)(k)$ would always reduce to
        $f(k - 0) = f(k)$, therefore $(0, \lambda\{\}) \mult (n, f) = (n, f)$.

        To show $\Array$ satisfies right unit,
        we want to show $(n, f) \mult (0, \lambda\{\}) = (n, f)$.
        \begin{align*}
            (n, f) \mult (0 , \lambda\{\}) & = (n + 0, f \oplus \lambda\{\})               \\
            (f \oplus \lambda\{\})(k)      & = \begin{cases}
                                                   f(k)                 & \text{if}\ k < n \\
                                                   (\lambda\{\})(k - 0) & \text{otherwise}
                                               \end{cases}
        \end{align*}
        It is trivial to see the length matches: $n + 0 = n$. We also need to show $f \oplus \lambda\{\} = f$.
        We note that the type of $f \oplus \lambda\{\}$ is $\Fin[n + 0] \to A$, therefore $k$ is of the type $\Fin[n + 0]$.
        Since $\Fin[n+0] \cong \Fin[n]$, it must always hold that $k < n$, therefore $(f \oplus \lambda\{\})(k)$ must
        always reduce to $f(k)$. Thus, $(n, f) \mult (0, \lambda\{\}) = (n, f)$.

        For associativity, we want to show for any array $(n, f)$, $(m, g)$, $(o, h)$,
        $((n, f) \mult (m, g)) \mult (o, h) = (n, f) \mult ((m, g) \mult (o, h))$.
        \begin{align*}
            ((n, f) \mult (m, g)) \mult (o, h) & = ((n + m) + o, (f \oplus g) \oplus h)                        \\
            ((f \oplus g) \oplus h)(k)         & = \begin{cases}
                                                       \begin{cases}
                                                           f(k)     & \text{if}\ k < n \\
                                                           g(k - n) & \text{otherwise}
                                                       \end{cases}
                                                                      & \text{if}\ k < n + m \\
                                                       h(k - (n + m)) & \text{otherwise}
                                                   \end{cases}                   \\
            (n, f) \mult ((m, g) \mult (o, h)) & = (n + (m + o), f \oplus (g \oplus h))                        \\
            (f \oplus (g \oplus h))(k)         & = \begin{cases}
                                                       f(k)                                   & \ \text{k < n} \\
                                                       \begin{cases}
                                                           g(k - n)     & \text{k - n < m} \\
                                                           h(k - n - m) & \text{otherwise} \\
                                                       \end{cases} & \text{otherwise}
                                                   \end{cases}
        \end{align*}
        We first case split on $k < n + m$ then $k < n$.

        Case $k < n + m$, $k < n$: Both $(f \oplus (g \oplus h))(k)$ and $((f \oplus g) \oplus h)(k)$ reduce to $f(k)$.

        Case $k < n + m$, $k \geq n$: $((f \oplus g) \oplus h)(k)$ reduce to $g(k - n)$ by definition.
        To show $(f \oplus (g \oplus h))(k)$ would also reduce to $g(k - n)$, we first need to show $\neg(k < n)$,
        which follows from $k \geq n$. We then need to show $k - n < m$.
        This can be done by simply subtracting $n$ from both side on $k < n + m$, which is well defined since $k \geq n$.

        Case $k \geq n + m$: $((f \oplus g) \oplus h)(k)$ reduce to $h(k - (n + m))$ by definition.
        To show $(f \oplus (g \oplus h))(k)$ would also reduce to $h(k - (n + m))$,
        we first need to show $\neg(k < n)$, which follows from $k \geq n + m$.
        We then need to show $\neg(k - n < m)$, which also follows from $k \geq n + m$.
        We now have $(f \oplus (g \oplus h))(k) = h(k - n - m)$. Since $k \geq n + m$, $h(k - n - m)$ is well defined
        and is equal to $h(k - (n + m))$, therefore $(f \oplus (g \oplus h))(k) = (f \oplus g) \oplus h)(k) = h(k - (n + m))$.
        In all cases $(f \oplus (g \oplus h))(k) = ((f \oplus g) \oplus h)(k)$, therefore associativity holds.
    \end{proof}
    For brevity, we write $S$ for the successor function on natural numbers.
\end{toappendix}
We need two auxiliary lemmas about array concatenation to do induction on arrays.
\begin{lemmarep}[\alink{lemma}{Array cons}{array-eta-suc}]
    \label{array:eta-suc}
    Any array $(S(n), f)$ is equal to $\eta_A (f(0)) \mult (n, f \comp S)$.
\end{lemmarep}

\begin{proof}
    We want to show $\eta_A (f(0)) \mult (n, f \comp S) = (S(n), f)$.
    \begin{align*}
        (1, \lambda\{ 0 \mapsto f(0) \}) \mult (n , f \comp S) &
        = (1 + n, \lambda\{ 0 \mapsto f(0) \} \oplus (f \comp S))                                            \\
        (\lambda\{ 0 \mapsto f(0) \} \oplus (f \comp S))(k)    & = \begin{cases}
                                                                       f(0)               & \text{if}\ k < 1 \\
                                                                       (f \comp S)(k - 1) & \text{otherwise}
                                                                   \end{cases}
    \end{align*}
    It is trivial to see the length matches: $1 + n = S(n)$. We need to show
    $(\lambda\{ 0 \mapsto f(0) \} \oplus (f \comp S)) = f$.
    We prove by case splitting on $k < 1$.
    On $k < 1$, $(\lambda\{ 0 \mapsto f(0) \} \oplus (f \comp S))(k)$ reduces to $f(0)$.
    Since, the only possible for $k$ when $k < 1$ is 0, $(\lambda\{ 0 \mapsto f(0) \} \oplus (f \comp S))(k) = f(k)$
    when $k < 1$.
    On $k \geq 1$, $(\lambda\{ 0 \mapsto f(0) \} \oplus (f \comp S))(k)$ reduces to $(f \comp S)(k - 1) = f(S(k - 1))$.
    Since $k \geq 1$, $S(k - 1) = k$, therefore $(\lambda\{ 0 \mapsto f(0) \} \oplus (f \comp S))(k) = f(k)$
    when $k \geq 1$.
    Thus, in both cases, $(\lambda\{ 0 \mapsto f(0) \} \oplus (f \comp S)) = f$.
\end{proof}

\begin{lemmarep}[\alink{lemma}{Array split}{array-split}]
    \label{array:split}
    For any array $(S(n), f)$ and $(m, g)$,
    \[
        (n + m, (f \oplus g) \comp S) = (n, f \comp S) \mult (m, g)
        \enspace .
    \]
\end{lemmarep}
\begin{proof}
    It is trivial to see both array have length $n + m$.
    We want to show that:
    \[
        (f \oplus g) \comp S = (f \comp S) \oplus g \enspace.
    \]
    \begin{align*}
        ((f \oplus g) \comp S)(k) & = \begin{cases}
                                          f(S(k))        & \text{if}\ S(k) < S(n) \\
                                          g(S(k) - S(n)) & \text{otherwise}
                                      \end{cases} \\
        ((f \comp S) \oplus g)(k) & = \begin{cases}
                                          (f \comp S)(k) & \text{if}\ k < n \\
                                          g(k - n)       & \text{otherwise}
                                      \end{cases}
    \end{align*}
    We prove by case splitting on $k < n$.
    On $k < n$, $((f \oplus g) \comp S)(k)$ reduces to $f(S(k))$ since $k < n$ implies $S(k) < S(n)$,
    and $((f \comp S) \oplus g)(k)$ reduces to $(f \comp S)(k)$ by definition, therefore they are equal.
    On $k \geq n$, $((f \oplus g) \comp S)(k)$ reduces to $g(S(k) - S(n)) = g(k - n)$,
    and $((f \comp S) \oplus g)(k)$ reduces to $g(k - n)$ by definition, therefore they are equal.
\end{proof}
\begin{toappendix}
    Informally, given a non-empty array $u$ and any array $v$,
    concatenating $u$ with $v$ then dropping the first element is the same as
    dropping the first element of $u$ then concatenating with $v$.
\end{toappendix}
\begin{definition}[\alink{definition}{Universal extension}{array-universal-extension}]
    Given a monoid $\mathfrak{X}$, and a map $f : A \to X$,
    we define $\ext{f} : \Array(A) \to X$, by induction on the length of the array:
    \begin{align*}
        \ext{f}(0 , g)    & = e                                    \\
        \ext{f}(S(n) , g) & = f(g(0)) \mult \ext{f}(n , g \circ S)
    \end{align*}
\end{definition}
\begin{toappendix}
    \begin{proposition}[\alink{proposition}{}{array-ext-lifts-homomorphism}]
        $\ext{(\blank)}$ lifts a function $f : A \to X$ to a monoid homomorphism $\ext{f} : \Array(A) \to \mathfrak{X}$.
    \end{proposition}

    \begin{proof}
        To show that $\ext{f}$ is a monoid homomorphism,
        we need to show $\ext{f}(\xs \mult \ys) = \ext{f}(\xs) \mult \ext{f}(\ys)$.
        We can do so by induction on $\xs$.

        Case $(0, g)$:
        We have $g = \lambda\{\}$ by~\cref{array:zero-is-id}.
        \begin{align*}
             & \mathrel{\phantom{=}} \ext{f}((0, \lambda\{\}) \mult \ys)                                            \\
             & = \ext{f}(\ys)                                            & \text{by definition of concatenation}    \\
             & = e \mult \ext{f}(\ys)                                    & \text{by left unit}                      \\
             & = \ext{f}(0, \lambda\{\}) \mult \ext{f}(\ys)              & \text{by definition of $\ext{(\blank)}$} \\
        \end{align*}

        Case $(S(n), g)$: Let $\ys$ be $(m, h)$.
        \begin{align*}
             & \mathrel{\phantom{=}} \ext{f}((S(n), g) \mult (m, h))                                                              \\
             & = \ext{f}(S(n + m), g \oplus h)                                 & \text{by definition of concatenation}            \\
             & = f((g \oplus h)(0)) \mult \ext{f}(n + m, (g \oplus h) \comp S) & \text{by definition of $\ext{(\blank)}$}         \\
             & = f(g(0)) \mult \ext{f}(n + m, (g \oplus h) \comp S)            & \text{by definition of $\oplus$, and $0 < S(n)$} \\
             & = f(g(0)) \mult \ext{f}((n, g \comp S) \mult (m, h))            & \text{by~\cref{array:split}}                     \\
             & = f(g(0)) \mult (\ext{f}(n, g \comp S) \mult \ext{f}(m, h)))    & \text{by induction}                              \\
             & = (f(g(0)) \mult \ext{f}(n, g \comp S)) \mult \ext{f}(m, h))    & \text{by associativity}                          \\
             & = \ext{f}(S(n), g) \mult \ext{f}(m, h))                         & \text{by definition of $\ext{(\blank)}$}         \\
        \end{align*}
        Therefore, $\ext{(\blank)}$ correctly lifts a function to a monoid homomorphism.
    \end{proof}
\end{toappendix}
\begin{propositionrep}[\alink{proposition}{Universal property for Array}{array-univ}]
    \label{array:univ}
    \leavevmode
    $(\Array(A),\eta_A)$ is the free monoid on $A$.
\end{propositionrep}
\begin{proofsketch}
    We need to show that $\ext{(\blank)}$ is an inverse to ${(\blank) \comp \eta_A}$.
    $\ext{f} \comp \eta_A = f$ for all set functions $f : A \to X$ holds trivially.
    To show $\ext{(f \comp \eta_A)} = f$ for all homomorphisms $f : \Array(A) \to \mathfrak{X}$,
    we need to show $\forall \xs.\, \ext{(f \comp \eta_A)}(\xs) = f(\xs)$.
    We use \cref{array:eta-suc,array:split} to do induction on $\xs$,
    similar to doing induction on the inductive datatype $\List(A)$.
\end{proofsketch}
\begin{proof}
    To show that $\ext{(\blank)}$ is an inverse to $\blank \comp \eta_A$,
    we first show $\ext{(\blank)}$ is the right inverse to $\blank \comp \eta_A$.
    For all $f$ and $x$, $(\ext{f} \circ \eta_A)(x) = \ext{f}(1, \lambda\{0 \mapsto x\}) = f(x) \mult e = f(x)$,
    therefore by function extensionality, for any $f$, $\ext{f} \circ \eta_A = f$,
    and $(\blank \circ \eta_A) \comp \ext{(\blank)} = id$.

    To show $\ext{(\blank)}$ is the left inverse to $\blank \comp \eta_A$, we need to prove
    for any monoid homomorphism $f : \Array(A) \to \mathfrak{X}$, $\ext{(f \comp \eta_A)}(\xs) = f(\xs)$.
    We can do so by induction on $\xs$.

    Case $(0, g)$:
    By~\cref{array:zero-is-id} we have $g = \lambda\{\}$.
    \begin{align*}
        \ext{(f \comp \eta_A)}(0, \lambda\{\})
         & = e                 & \text{by definition of $\ext{(\blank)}$} \\
         & = f(0, \lambda\{\}) & \text{by homomorphism properties of $f$}
    \end{align*}

    Case $(S(n), g)$, we prove it in reverse:
    \begin{align*}
         & \mathrel{\phantom{=}} f(S(n), g)                                                                               \\
         & = f(\eta_A(g(0)) \mult (n, g \comp S))                              & \text{by~\cref{array:eta-suc}}           \\
         & = f(\eta_A(g(0))) \mult f(n, g \comp S)                             & \text{by homomorphism properties of $f$} \\
         & = (f \comp \eta_A)(g(0)) \mult \ext{(f \comp \eta_A)}(n, g \comp S) & \text{by induction}                      \\
         & = \ext{(f \comp \eta_A)}(S(n), g)                                   & \text{by definition of $\ext{(\blank)}$}
    \end{align*}
    By function extensionality, $\ext{(\blank)} \comp (\blank \circ \eta_A) = id$.
    Therefore, $\ext{(\blank)}$ and $(\blank) \circ [\_]$ are inverse of each other.

    We have now shown that $(\blank) \comp \eta_A$ is an equivalence from
    monoid homomorphisms $\Array(A) \to \mathfrak{X}$
    to set functions $A \to X$, and its inverse is given by $\ext{(\blank)}$, which maps set
    functions $A \to X$ to monoid homomorphisms $\Array(A) \to \mathfrak{X}$. Therefore, $\Array$ is indeed
    the free monoid.
\end{proof}
An alternative proof of the universal property for $\Array$ can be given by directly constructing an equivalence (of
types, and monoid structures) between $\Array(A)$ and $\List(A)$ (using $\term{tabulate}$ and $\term{lookup}$), and then
applying univalence and transport (see formalisation).

\section{Constructions of Free Commutative Monoids}
\label{sec:commutative-monoids}

The next step is to add commutativity to the monoid multiplication in each construction of free monoids.
Informally, adding commutativity to free monoids turns ``ordered lists'' to ``unordered lists'',
where the ordering is the one imposed by the position or index of the elements in the list.
This is crucial to our goal of studying sorting,
as we will study sorting as a function mapping back unordered lists to ordered lists,
later in~\cref{sec:sorting}.

It is known that finite multisets are (free) commutative monoids,
under the operation of multiset union: $\xs \cup \ys = \ys \cup \xs$.
The order is ``forgotten'' in the sense that it doesn't matter how two multisets are union-ed together,
such as, $\bag{a, a, b, c}$ and $\bag{b, a, c, a}$ are equal as finite multisets
(justifying the bag notation).
Formally, $\bag{x,y,\dots}$ denotes $\eta_A(x) \mult \eta_A(y) \mult \dots : \MM(A)$.

This is unlike free monoids,
where $[a, a, b, c] \neq [b, a, c, a]$ (justifying the list notation).
Formally, $[x, y, \dots]$ denotes $\eta_A(x) \mult \eta_A(y) \mult \dots : \LL(A)$,
and $x \cons xs$ denotes $\eta_A(x) \mult xs : \LL(A)$.

\subsection{Free monoids quotiented by permutation relations}
\label{sec:cmon:qfreemon}

Instead of constructing free commutative monoids directly,
the first construction we study is to take \emph{any} free monoid and quotient by \emph{symmetries}.
The constructions in~\cref{sec:cmon:plist} and~\cref{sec:cmon:bag} are specific instances of this recipe.

From the perspective of universal algebra and equational logic developed in~\cref{sec:universal-algebra},
we work in the signature of monoids, but extend the equational theory of monoids with commutativity.
If $(\str{F}(A), \eta)$ is a free monoid construction satisfying its universal property,
then we quotient $F(A)$ by an \emph{appropriate} symmetry relation $\approx$.
This is precisely a \emph{permutation relation}.

\begin{definition}[\alink{definition}{Permutation relation}{permutation-relation}]
    \label{def:permutation-relation}
    \leavevmode
    A binary relation on a free monoid $(F(A), e, \mult)$ is a permutation relation iff it is:
    \begin{itemize}
        \item reflexive, symmetric, transitive (an equivalence),
        \item a congruence w.r.t $\mult$: $\forall a, b, c, d, \, a \approx b \to c \approx d \to a \mult c \approx b \mult d$,
        \item commutative: $\forall a, b, \, a \mult b \approx b \mult a$, and
        \item respects $\ext{(\blank)}$: $\forall a, b, f, \, a \approx b \to \ext{f}(a) = \ext{f}(b)$.
    \end{itemize}
\end{definition}
Note that being a permutation relation is a proposition.
We write $\quot{F(A)}{\approx}$ for the quotient of $F(A)$ by $\approx$,
with the inclusion map $q: F(A) \twoheadrightarrow \quot{F(A)}{\approx}$.
The generators map into $\quot{F(A)}{\approx}$ is given by $q \comp \eta_A$,
the identity element is $\qlift{e} = q(e)$,
and the multiplication operation $\mult$ is lifted to $\qlift{\mult}$ in the quotient by congruence.
\begin{propositionrep}[\alink{proposition}{}{qfreemon-cmonoid}]
    $(\quot{\str{F}(A)}{\approx}, \qlift{e}, \qlift{\mult})$ is a commutative monoid.
\end{propositionrep}
\begin{proof}
    Since $\approx$ is a congruence w.r.t. $\mult$,
    we can lift $\mult : F(A) \to F(A) \to F(A)$ to the quotient to obtain
    $\qlift{\mult} : \quot{F(A)}{\approx} \to \quot{F(A)}{\approx} \to \quot{F(A)}{\approx}$.
    $\qlift{\mult}$ also satisfies unit and associativity laws on quotiented elements,
    because these laws hold on representatives in $F(A)$.
    Commutativity of $\qlift{\mult}$ follows from the commutativity requirement of $\approx$.
    This makes $(\quot{F(A)}{\approx}, \qlift{e}, \qlift{\mult})$ a commutative monoid.
    Note also that $q(a \mult b) = q(a) \qlift{\mult} q(b)$.
\end{proof}
The monoid extension operation of $F(A)$, denoted $\ext{(\blank)}$,
lifts to an extension operation for commutative monoids, for $\quot{F(A)}{\approx}$,
which we denote by $\exthat{(\blank)}$.
\begin{definition}[\alink{definition}{}{qfreemon-ext}]
    Given a commutative monoid $\str{X}$ and a map $f : A \to X$,
    we define the extension operation as follows:
    we first obtain $\ext{f} : \str{F}(A) \to \str{X}$ by the universal property of $F$,
    then lift it to $\exthat{f} : \quot{\str{F}(A)}{\approx} \; \to \str{X}$,
    by using that $\approx$ respects $\ext{(\blank)}$, as in~\cref{def:permutation-relation}.
\end{definition}

\begin{propositionrep}[\alink{proposition}{Universal property for $\quot{\str{F}(A)}{\approx}$}{qfreemon-univ}]
    \label{prop:qfreemon}
    \leavevmode
    $(\quot{\str{F}(A)}{\approx}, \eta_A : {A \xto{\eta_A} \str{F}(A) \xto{q} \quot{\str{F}(A)}{\approx}})$
    with the extension operation $\exthat{(\blank)}$ forms the free commutative monoid on $A$.
\end{propositionrep}
\begin{proof}
    We have the following situation.
    \[\begin{tikzcd}[ampersand replacement=\&,cramped]
            {\quot{F(A)}{\approx}} \&\& X \\
            \\
            {F(A)} \&\& X \\
            \\
            A
            \arrow["{\exthat{f}}", from=1-1, to=1-3]
            \arrow["q", two heads, from=3-1, to=1-1]
            \arrow["{\ext{f}}", from=3-1, to=3-3]
            \arrow[equals, from=3-3, to=1-3]
            \arrow["{\eta_{A}}", from=5-1, to=3-1]
            \arrow["f"', from=5-1, to=3-3]
        \end{tikzcd}\]
    By definition of $\exthat{(\blank)}$, the upper square commutes, and so the whole diagram commutes.

    To see that $\exthat{(\blank)}$ is the right inverse to $(\blank) \comp q \comp \eta_A$,
    observe that $\exthat{f} \comp q \comp \eta_A = \ext{f} \comp \eta_A = f$.
    And,
    $\exthat{f}(q(x) \qlift{\mult} q(y))
        = \exthat{f}{(q(x \mult y))}
        = \ext{f}(x \mult y) = \ext{f}(x) \mult \ext{f}(y)
        = \exthat{f}(q(x)) \mult \exthat{f}(q(y))
    $,
    making $\exthat{f}$ a (commutative) monoid homomorphism.

    To show $\exthat{(\blank)}$ is the left inverse to $(\blank) \comp q \comp \eta_A$,
    we need to show for any commutative monoid homomorphism
    $f : \quot{\str{F}(A)}{\approx} \to \str{X}$, that $\exthat{(f \comp q \comp \eta_A)} = f$.
    It is enough to check on generators, so for any $x : F(A)$,
    we note that $\exthat{(f \comp q \comp \eta_A)}(q(x)) = \ext{(f \comp q \comp \eta_A)}(x) = f(q(x))$
    since $f \comp q$ is a monoid homomorphism, and we are done.

    We have now shown that $(\blank) \comp q \comp \eta_A$ produces an equivalence from
    commutative monoid homomorphisms $\quot{\str{F}(A)}{\approx} \to \str{X}$
    to set functions $A \to X$, and its inverse is given by $\exthat{(\blank)}$, which maps set
    functions $A \to X$ to commutative monoid homomorphisms $\quot{\str{F}(A)}{\approx} \to \str{X}$.
    Therefore, $\quot{\str{F}(A)}{\approx}$ is indeed the free commutative monoid on $A$
    with generators map $q \comp \eta_A$ and extension operation $\exthat{(\blank)}$.
\end{proof}

\subsection{Lists quotiented by permutation relations}
\label{sec:cmon:plist}

A specific instance of this construction is $\List$ quotiented by a concrete permutation relation to get commutativity.
Of course, there are many permutation relations in the literature,
we consider a simple one which swaps any two adjacent elements anywhere in the middle of the list.
This $\PList$ construction has also been considered in~\cite{joramConstructiveFinalSemantics2023},
who prove that $\PList$ is equivalent to the free commutative monoid (constructed as a HIT).
We give a direct proof of its universal property using our generalisation.
\begin{definition}[\alink{definition}{Perm}{perm}]
    The inductive family $\Perm$ is generated by two constructors
    and defines a permutation relation on $\List$:
    \begin{align*}
        \term{perm-refl} & : \forall \xs,\, \Perm\,\xs\,\xs \\
        \term{perm-swap} & : \forall x, y, \xs, \ys, \zs,\,
        \Perm\,(\xs \mult x \cons y \cons \ys)\,\zs
        \to \Perm\,(\xs \mult y \cons x \cons \ys)\,\zs
    \end{align*}
    where $\mult$ is the usual list concatenation operation.
    $\PList(A)$ is the quotient of $\List(A)$ by $\Perm$:
    \[
        \PList(A) \defeq \quot{\List(A)}{\Perm}
        \enspace.
    \]
\end{definition}
By~\cref{sec:cmon:qfreemon}, to obtain that $\PList(A)$ is the free commutative monoid on $A$, it suffices to prove the property of $\Perm(A)$ being a permutation relation.

\begin{propositionrep}[\alink{proposition}{Perm is a permutation relation}{perm-is-perm-relation}]
    \label{prop:perm-is-perm-relation}
    The relation $\Perm$ is a permutation relation on $\List(A)$.
\end{propositionrep}
\begin{proof}
    We must show that $\Perm$ is:
    \begin{enumerate}
        \item an equivalence relation (reflexive, symmetric, transitive): reflexivity is by construction ($\term{perm-refl}$). Symmetry and transitivity are straightforward by induction on the relation.
        \item a congruence w.r.t $\mult$: we want to show for any lists $\xs, \ys, \zs, \ws$, if $\xs \approx \ys$ and $\zs \approx \ws$, then $\xs \mult \zs \approx \ys \mult \ws$, which is proven by induction.
        \item commutative: $\xs \mult \ys \approx \ys \mult \xs$, proven by induction on the lists.
        \item respects $\ext{(\blank)}$: shown in~\cref{plist:sharp-sat}.
    \end{enumerate}
    Thus, $\Perm$ is a permutation relation.
\end{proof}

\begin{propositionrep}[\alink{proposition}{}{plist-resp-ext}]
    \label{plist:sharp-sat}
    Let $\str{X}$ be a commutative monoid, and $f : A \to X$.
    For $x,y : A$ and $\xs, \ys : \List(A)$,
    $\ext{f}(\xs\,\mult\,x \cons y \cons \ys) \id \ext{f}(\xs\,\mult\,y \cons x \cons \ys)$.
    Hence, $\Perm$ respects $\ext{(\blank)}$.
\end{propositionrep}
\begin{proof}
    We can prove this by induction on $\xs$. For $\xs = []$, by homomorphism properties of $\ext{f}$,
    we can prove $\ext{f}(x \cons y \cons \ys) = \ext{f}([ x ]) \mult \ext{f}([ y ]) \mult \ext{f}(\ys)$.
    Since the image of $\ext{f}$ is a commutative monoid, we have
    $\ext{f}([ x ]) \mult \ext{f}([ y ]) = \ext{f}([ y ]) \mult \ext{f}([ x ])$, therefore proving
    $\ext{f}(x \cons y \cons \ys) = \ext{f}(y \cons x \cons \ys)$. For $\xs = z \cons \zs$, we can prove
    $\ext{f}((z \cons \zs)\,\mult\,x \cons y \cons \ys) = \ext{f}([ z ]) \mult \ext{f}(\zs\,\mult\,x \cons y \cons \ys)$.
    We can then complete the proof by induction to obtain
    $\ext{f}(\zs\,\mult\,x \cons y \cons \ys) = \ext{f}(\zs\,\mult\,y \cons x \cons \ys)$ and reassembling
    back to $\ext{f}((z \cons \zs)\,\mult\,y \cons x \cons \ys)$ by homomorphism properties of $\ext{f}$.
\end{proof}
$\PList$ being a quotient of $\List$ makes it easy to
lift functions and properties defined on $\List$ along the quotient map.
The inductive nature of $\Perm$ makes it easy to define algorithms that inductively rearrange elements in a list,
such as insertion sort.
This property also makes it possible to normalize inductively constructed list elements in $\PList$
to their simplest forms, e.g. $q([ x ]) \mult q([y, z])$ normalizes to $q([x,y,z])$ by definition,
saving the efforts of defining auxiliary lemmas to prove their equality.

This inductive nature, however, makes it difficult to define efficient operations on $\PList$. Consider a
divide-and-conquer algorithm such as merge sort, which involves partitioning a $\PList$ of length $n+m$ into
two smaller $\PList$ of length $n$ and length $m$.
We must iterate all $n$ elements before we can make such a partition, thus making $\PList$ unintuitive
to work with when we want to reason with operations that involve arbitrary partitioning.
\begin{toappendix}
    Whenever we define a function on $\PList$ by pattern matching we would also need to show
    the function respects $\Perm$, i.e. $\Perm \as\,\bs \to f(\as) = f(\bs)$. This can be annoying because
    of the many auxiliary variables in the constructor $\term{perm-swap}$, namely $\xs$, $\ys$, $\zs$.
    We need to show $f$ would respect a swap in the list anywhere between $\xs$ and $\ys$, which can
    unnecessarily complicate the proof. For example in the inductive step of~\cref{plist:sharp-sat},
    $\ext{f}((z \cons \zs)\,\mult\,x \cons y \cons \ys) = \ext{f}([ z ]) \mult \ext{f}(\zs\,\mult\,x \cons y \cons \ys)$,
    to actually prove this in Cubical Agda would involve first applying associativity to prove
    $(z \cons \zs)\,\mult\,x \cons y \cons \ys = z \cons (\zs\,\mult\,x \cons y \cons \ys)$, before we can actually
    apply homomorphism properties of $f$. In the final reassembling step, similarly,
    we also need to re-apply associativity to go from $z \cons (\zs\,\mult\,y \cons x \cons \ys)$
    to $(z \cons \zs)\,\mult\,y \cons x \cons \ys$. Also since we are working with an equivalence relation we
    defined ($\Perm$) and not the equality type directly, we cannot exploit the many combinators defined
    in the standard library for the equality type and often needing to re-define combinators ourselves.
\end{toappendix}

\subsection{Swap Lists}
\label{sec:cmon:slist}

Informally, quotients are defined by generating all the points, then quotienting out into equivalence classes by the
congruence relation.
Alternately, HITs use generators (points) and higher generators (paths) (and higher higher generators and so on\ldots).
We can also define free commutative monoids using HITs where adjacent swaps generate all symmetries,
for example swap-lists taken from \cite{choudhuryFreeCommutativeMonoids2023},
and in the Cubical library~\cite{theagdacommunityCubicalAgdaLibrary2025}.
\begin{definition}[\alink{definition}{$\SList$}{slist}]
    \label{def:slist}
    The higher inductive type $\SList(A)$ is generated by:
    \begin{align*}
        \nil              & : \SList(A)                                                   \\
        \blank\cons\blank & : A \to \SList(A) \to \SList(A)                               \\
        \term{swap}       & : \forall x, y, xs,\, x \cons y \cons xs = y \cons x \cons xs \\
        \term{trunc}      & : \forall x, y,\, (p, q : x = y) \to p = q
    \end{align*}
\end{definition}
\begin{toappendix}
    \begin{definition}[Concatenation]
        We define the concatenation operation $\mult : \SList(A) \to \SList(A) \to \SList(A)$
        recursively, where we also have to consider the (functorial) action on the $\term{swap}$ path using $\term{ap}$.
        \begin{align*}
            [] \mult \ys                                  & \defeq \ys                         \\
            (x \cons \xs) \mult \ys                       & \defeq x \cons (\xs \mult \ys)     \\
            \term{ap}_{\mult \ys}(\term{swap}(x, y, \xs)) & = \term{swap}(x, y, \ys \mult \xs)
        \end{align*}
    \end{definition}
    A proof of $(\SList(A), \mult, [])$ satisfying commutative monoid laws
    is given in~\cite{choudhuryFreeCommutativeMonoids2023} and the Cubical Agda library~\cite{theagdacommunityCubicalAgdaLibrary2025}.
    We explain the proof of $\mult$ satisfying commutativity here.

    \begin{lemma}[Head rearrange]\label{slist:cons}
        For all $x : A$, $\xs : \SList(A)$, $x \cons \xs = \xs \mult [ x ]$.
    \end{lemma}

    \begin{proof}
        We can prove this by induction on $\xs$.
        For $\xs \jdgeq []$ this is trivial. For $\xs \jdgeq y \cons \ys$, we have the induction hypothesis $x \cons \ys = \ys \mult [ x ]$.
        By applying $y \cons (\blank)$ on both sides and then applying a $\term{swap}$, we complete the proof.
    \end{proof}

    \begin{theorem}[Commutativity]\label{slist:comm}
        For all $\xs,\,\ys : \SList(A)$, $\xs \mult \ys = \ys \mult \xs$.
    \end{theorem}

    \begin{proof}
        By induction on $\xs$ we can iteratively apply~\cref{slist:cons} to move all elements of $\xs$
        to after $\ys$. This moves $\ys$ to the head and $\xs$ to the end, thereby proving
        $\xs \mult \ys = \ys \mult \xs$.
    \end{proof}
\end{toappendix}
We refer the reader to~\cite{choudhuryFreeCommutativeMonoids2023} or the Cubical Agda library~\cite{theagdacommunityCubicalAgdaLibrary2025}
for further details.
Much like $\List$, $\SList$ is inductively defined, therefore making it intuitive to reason
with when defining inductive operations or proofs on $\SList$, but difficult to reason with
when defining operations that involve arbitrary partitioning, for reasons similar to those given
in~\cref{sec:cmon:plist}.

\begin{toappendix}
    Unlike $\PList$ which is defined as a set quotient, $\SList$ is defined as a HIT, therefore handling equalities
    between $\SList$ is much simpler than $\PList$. We would still need to prove a function $f$ respects
    the path constructor of $\SList$ when pattern matching, i.e. $f(x \cons y \cons \xs) = f(y \cons x \cons \xs)$.
    Unlike $\PList$ we do not need to worry about as many auxiliary variables since swap
    only happens at the head of the list, whereas with $\PList$ we would need to prove
    $f(\xs\,\mult\,x \cons y \cons \ys) = f(\xs\,\mult\,y \cons x \cons \ys)$. One may be tempted to just remove $\xs$
    from the $\term{perm-swap}$ constructor such that it becomes
    $\term{perm-swap} : \Perm\,(x \cons y \cons \ys)\,\zs \to \Perm\,(y \cons x \cons \ys)\,\zs$.
    However this would break $\Perm$'s congruence w.r.t. $\mult$, therefore violating the axioms of
    permutation relations. Also, since we are working with the identity type directly, properties we would
    expect from $\term{swap}$, such as reflexivity, transitivity, symmetry, congruence and such all comes directly by
    construction, whereas with $\Perm$ we would have to prove these properties manually.
    We can also use the many combinatorics defined in the standard library for equational reasoning,
    making the handling of $\SList$ equalities a lot simpler.
\end{toappendix}

\subsection{Bag}
\label{sec:cmon:bag}

Alternatively, we can also quotient the $\Array$ type by symmetries to get commutativity.
This construction has been considered in~\cite{altenkirchDefinableQuotientsType2011}
and~\cite{liQuotientTypesType2015},
partially considered in~\cite{choudhuryFreeCommutativeMonoids2023},
and in~\cite{joramConstructiveFinalSemantics2023},
who give a similar construction, where only the index function is quotiented, instead of
the entire array.
\cite{danielssonBagEquivalenceProofRelevant2012} also considered $\Bag$ as a setoid relation
on $\List$ in an intensional MLTT setting.
\cite{joramConstructiveFinalSemantics2023} prove that their version of $\Bag$
is the free commutative monoid by equivalence to the other HIT constructions.
We give a direct proof of its universal property instead, using our general recipe.

\begin{definition}[\alink{definition}{Bag}{bag}]
    \label{def:bag}
    Bags are defined as arrays quotiented by bag equivalence $\approx$:
    \begin{align*}
        (n , f) \approx (m , g) & \defeq \dsum{\sigma : \Fin[n] \xto{\sim} \Fin[m]} (f = g \comp \sigma)
        \\
        \Bag(A)                 & \defeq \quot{\Array(A)}{\approx}
    \end{align*}
\end{definition}
Note that by a pigeonhole argument,
$\sigma$ may only be constructed when $n = m$.
In other words, we are quotienting by an automorphism on the indices,
and its action on the elements.
We have already shown $\Array$ to be the free monoid in~\cref{sec:arrays}.
By~\cref{prop:qfreemon}, to obtain that $\Bag(A)$ is the free commutative monoid on $A$, it suffices to prove the property of $\approx$ being a permutation relation.

\begin{propositionrep}[\alink{proposition}{Bag equivalence is a permutation relation}{bag-is-perm-relation}]
    \label{prop:bag-is-perm-relation}
    The relation $\approx$ is a permutation relation on $\Array(A)$.
\end{propositionrep}
\begin{proof}
    By definition of a permutation relation (\cref{def:permutation-relation}), we must show that $\approx$ is:
    an equivalence relation (proven in~\cref{bag-equiv}),
    a congruence w.r.t $\mult$ (proven in~\cref{bag:cong}),
    commutative (proven in~\cref{bag:comm}), and
    respects $\ext{(\blank)}$ (proven in~\cref{bag:perm-sat}).
    Thus, $\approx$ is a permutation relation.
\end{proof}
\begin{toappendix}
    \label{bag-equiv}
    \begin{proposition}[\alink{proposition}{}{bag-equiv}]
        $\approx$ is a equivalence relation.
    \end{proposition}
    \begin{proof}
        We can show any array $\xs$ is related to itself by the identity isomorphism, therefore $\approx$ is reflexive.
        If $\xs \approx \ys$ by $\sigma$, we can show $\ys \approx \xs$ by $\sigma^{-1}$, therefore $\approx$ is symmetric.
        If $\xs \approx \ys$ by $\sigma$ and $\ys \approx \zs$ by $\phi$, we can show $\xs \approx \zs$ by $\sigma \comp \phi$,
        therefore $\approx$ is transitive.
    \end{proof}
\end{toappendix}
\begin{proposition}[\alink{proposition}{}{bag-cong}]
    \label{bag:cong}
    $\approx$ is congruent with respect to $\mult$.
\end{proposition}
\begin{proof}
    Given $(n, f) \approx (m, g)$ by $\sigma$ and $(u, p) \approx (v, q)$ by $\phi$,
    we want to show $(n, f) \mult (u, p) \approx (m, g) \mult (v, q)$ by some $\tau$.
    We construct $\tau$ as follows:
    \[
        \tau \defeq \Fin[n+u] \xto{\sim} \Fin[n] + \Fin[u] \xto{\sigma + \phi} \Fin[m] + \Fin[v] \xto{\sim} \Fin[m+v] \\
    \]
    which operationally performs:
    \[
        \begin{tikzcd}[ampersand replacement=\&,cramped]
            {\{\color{red}0,1,\dots,n-1, \color{blue} n,n+1,\dots,n+u-1 \color{black}\}} \\
            { \{\color{red}\sigma(0),\sigma(1)\dots,\sigma(n-1), \color{blue}\phi(0),\phi(1),\dots,\phi(u-1) \color{black}\}}
            \arrow["{\sigma+\phi}", maps to, from=1-1, to=2-1]
        \end{tikzcd}
        \enspace.
    \]
\end{proof}
\begin{proposition}[\alink{proposition}{}{bag-comm}]
    \label{bag:comm}
    $\approx$ respects commutativity.
\end{proposition}

\begin{proof}
    We want to show for any arrays $(n, f)$ and $(m, g)$, $(n, f) \mult (m, g) \approx (m, g) \mult (n, f)$
    by some $\phi$.
    We use formal combinators (see~\cite{choudhurySymmetriesReversibleProgramming2022}) to define $\phi$:
    \[
        \phi \defeq \Fin[n+m] \xto{\sim} \Fin[n] + \Fin[m] \xto{\term{swap}_{+}} \Fin[m] + \Fin[n] \xto{\sim} \Fin[m+n] \\
    \]
    which operationally performs:
    \[
        \begin{tikzcd}[ampersand replacement=\&,cramped]
            {\{\color{red}0,1,\dots,n-1, \color{blue} n,n+1,\dots,n+m-1 \color{black}\}} \\
            { \{\color{blue}n,n+1\dots,n+m-1, \color{red}0,1,\dots,n-1 \color{black}\}}
            \arrow["\phi", maps to, from=1-1, to=2-1]
        \end{tikzcd}
        \enspace.
    \]
\end{proof}
To show that $\ext{f}$ is invariant under permutation: for all $\phi\colon \Fin[n]\xto{\sim}\Fin[n]$,
$\ext{f}(n, i) \id \ext{f}(n, i \circ \phi)$, we need some formal combinators for
\emph{punching in} and \emph{punching out} indices. These operations are
borrowed from~\cite{mozlerCubicalAgdaSimple2021} and
developed further
in~\cite{choudhurySymmetriesReversibleProgramming2022} for studying permutation codes.
\begin{lemma}[\alink{lemma}{}{bag-tau}]
    \label{bag:tau}
    Given $\phi\colon \Fin[S(n)]\xto{\sim}\Fin[S(n)]$, there is a permutation $\tau\colon \Fin[S(n)]\xto{\sim}\Fin[S(n)]$
    such that $\tau(0) = 0$, and $\ext{f}(S(n), i \circ \phi) = \ext{f}(S(n), i \circ \tau)$.
\end{lemma}

\begin{proof}
    Let $k$ be $\phi^{-1}(0)$, and $k + j = S(n)$, we construct $\tau$:
    \[
        \tau \defeq \Fin[S(n)] \xto{\phi} \Fin[S(n)] \xto{\sim} \Fin[k+j] \xto{\sim} \Fin[k] + \Fin[j]
        \xto{\term{swap}_{+}} \Fin[j] + \Fin[k] \xto{\sim} \Fin[j+k] \xto{\sim} \Fin[S(n)]
    \]
    \[
        \begin{tikzcd}[ampersand replacement=\&,cramped]
            {\{\color{blue}0, 1, 2, \dots, \color{red}k, k+1, k+2, \dots \color{black}\}} \\
            {\{\color{blue}x, y, z, \dots, \color{red}0, u, v, \dots \color{black}\}}
            \arrow["\phi", maps to, from=1-1, to=2-1]
        \end{tikzcd}
        \\
        \begin{tikzcd}[ampersand replacement=\&,cramped]
            {\{\color{blue}0, 1, 2, \dots, \color{red}k, k+1, k+2, \dots \color{black}\}} \\
            {\{\color{red}0, u, v, \dots, \color{blue}x, y, z, \dots \color{black}\}}
            \arrow["\tau", maps to, from=1-1, to=2-1]
        \end{tikzcd}
    \]
    It is trivial to show $\ext{f}(S(n), i \circ \phi) = \ext{f}(S(n), i \circ \tau)$, since the only
    operation on indices in $\tau$ is $\term{swap}_{+}$. It suffices to show $(S(n), i \circ \phi)$
    can be decomposed into two arrays such that $(S(n), i \circ \phi) = (k, g) \mult (j, h)$
    for some $g$ and $h$. Since the image of $\ext{f}$ is a commutative monoid, and $\ext{f}$ is a homomorphism,
    $\ext{f}((k, g) \mult (j, h)) = \ext{f}(k, g) \mult \ext{f}(j, h) = \ext{f}(j, h) \mult \ext{f}(k, g) =
        \ext{f}((j, h) \mult (k, g))$, thereby proving $\ext{f}(S(n), i \circ \phi) = \ext{f}(S(n), i \circ \tau)$.

\end{proof}

\begin{lemma}[\alink{lemma}{}{bag-punch}]
    \label{bag:punch}
    Given $\tau\colon \Fin[S(n)]\xto{\sim}\Fin[S(n)]$ where $\tau(0) = 0$,
    there is a $\psi : \Fin[n] \xto{\sim} \Fin[n]$ such that $\tau \circ S = S \circ \psi$.
\end{lemma}

\begin{proof}
    We construct $\psi$ as $\psi(x) = \tau(S(x)) - 1$.
    Since $\tau$ maps only 0 to 0 by assumption, $\forall x. \, \tau(S(x)) > 0$, therefore
    the $(- 1)$ is well defined. This corresponds to the special case for $k = 0$ of the standard
    punch-in and punch-out operations from the Agda standard library~\cite{Daggitt_The_Agda_standard_2025}, as used for Lehmer codes
    in~\cite{choudhurySymmetriesReversibleProgramming2022}.
    \[
        \begin{tikzcd}[ampersand replacement=\&,cramped]
            {\{\color{blue}0, \color{red}1, 2, 3, \dots \color{black}\}} \\
            { \{\color{blue}0, \color{red} x, y, z \dots \color{black}\}}
            \arrow["\tau", maps to, from=1-1, to=2-1]
        \end{tikzcd}
        \\
        \begin{tikzcd}[ampersand replacement=\&,cramped]
            {\{\color{red}0, 1, 2, \dots \color{black}\}} \\
            { \{\color{red} x-1, y-1, z-1 \dots \color{black}\}}
            \arrow["\psi", maps to, from=1-1, to=2-1]
        \end{tikzcd}
    \]
\end{proof}

\begin{theorem}[\alink{theorem}{Permutation invariance}{bag-perm-sat}]
    \label{bag:perm-sat}
    For all $\phi\colon \Fin[n]\xto{\sim}\Fin[n]$, $\ext{f}(n, i) \id \ext{f}(n, i \circ \phi)$.
\end{theorem}

\begin{proof}
    By induction on $n$.
    \begin{itemize}
        \item At $n = 0$, $\ext{f}(0, i) \id \ext{f}(0, i \circ \phi) = e$.
        \item At $n = S(m)$,
              \begin{align*}
                   & \mathrel{\phantom{=}} \ext{f}(S(m), i \circ \phi)                                                 \\
                   & = \ext{f}(S(m), i \circ \tau)                          & \text{by~\cref{bag:tau}}                 \\
                   & = f(i(\tau(0))) \mult \ext{f}(m, i \circ \tau \circ S) & \text{by definition of $\ext{(\blank)}$} \\
                   & = f(i(0)) \mult \ext{f}(m, i \circ \tau \circ S)       & \text{by construction of $\tau$}         \\
                   & = f(i(0)) \mult \ext{f}(m, i \circ S \circ \psi)       & \text{by~\cref{bag:punch}}               \\
                   & = f(i(0)) \mult \ext{f}(m, i \circ S)                  & \text{induction}                         \\
                   & = \ext{f}(S(m), i)                                     & \text{by definition of $\ext{(\blank)}$}
              \end{align*}
    \end{itemize}
\end{proof}
Unlike $\PList$ and $\SList$, $\Bag$ and its underlying construction $\Array$ are not inductively defined,
making it difficult to do induction on them. For example,
in the proof of~\cref{array:univ}, both~\cref{array:eta-suc,array:split} are needed to do
induction on $\Array$, as opposed to $\List$ and its quotients, where we can do induction simply by
pattern matching. Much like $\PList$, when defining functions on $\Bag$, we need to show they respect
$\approx$, i.e. $\as \approx \bs \to f(\as) = f(\bs)$. Notably, this is much more difficult than
$\PList$ or $\SList$ -- with $\PList$ and $\SList$ we only need to consider swapping adjacent elements,
while with $\Bag$ we need to consider all possible permutations. For example,
in the proof of~\cref{bag:perm-sat}, we need to first construct a $\tau$ which satisfies $\tau(0) = 0$ and prove
$\ext{f}(n, i \comp \sigma) = \ext{f}(n, i \comp \tau)$ before we can apply induction.

\begin{toappendix}
    Since $\Array$ and $\Bag$ are not simple data types, the definition of
    the monoid operation on them $\mult$ are necessarily more complicated, and unlike $\List$, $\PList$
    and $\SList$, constructions of $\Array$ and $\Bag$ via $\mult$ often would not normalize into a
    very simple form, but would instead expand into giant trees of terms. This makes it such that when working
    with $\Array$ and $\Bag$ we need to be very careful or otherwise Agda would be stuck trying to display
    the normalized form of $\Array$ and $\Bag$ in the goal and context menu. Type-checking also becomes a lengthy
    process that tests if the user possesses the virtue of patience.

    However, performing arbitrary partitioning with $\Array$ and $\Bag$ is much easier than
    $\List$, $\SList$, $\PList$. For example,
    one can simply use the combinator ${\Fin[n+m] \xto{\sim} \Fin[n] + \Fin[m]}$ to partition the array,
    then perform operations on the partitions, such as swapping in~\cref{bag:comm},
    or perform operations on the partitions individually, such as two individual permutations in~\cref{bag:cong}.
    This makes it so that when defining divide-and-conquer algorithms like merge sort,
    $\Bag$ and $\Array$ are more natural representations to work with than $\List$, $\SList$, or $\PList$.
\end{toappendix}

\section{Sorting Functions}%
\label{sec:application}

We will now put to work the universal properties of our types of (ordered) lists and unordered lists,
to define operations on them systematically, which are mathematically sound, and then reason about them.
First, we explore definitions of various operations on both free monoids and free commutative monoids.
By univalence (and the structure identity principle~\cite{aczelVoevodskysUnivalenceAxiom2011,univalentfoundationsprogramHomotopyTypeTheory2013}),
everything henceforth holds for any presentation of free monoids and free commutative monoids.
We use $\FF(A)$ to denote the free monoid or free commutative monoid on $A$,
$\LL(A)$ to exclusively denote the free monoid construction,
and $\MM(A)$ to exclusively denote the free commutative monoid construction.

For example $\term{length}$ is a common operation defined inductively for $\List$,
but usually, properties about $\term{length}$, such as,
$\term{length}(\xs \mult \ys) = \term{length}(\xs) + \term{length}(\ys)$,
are proven separately after defining it.
In our framework of free algebras, where the $\ext{(\blank)}$ operation is a correct-by-construction homomorphism,
we can define operations like $\term{length}$ directly by universal extension,
which also gives us a proof that they are homomorphisms for free.
A further application of the universal property is to prove that two different types (such as lists and arrays) are equal,
by showing they both satisfy the same universal property as in~\cref{lem:free-algebras-unique},
which is desirable especially when proving a direct equivalence between the two types turns out
to be a difficult exercise in combinatorics.


\subsection{Prelude}%
\label{sec:prelude}

Any presentation of free monoids or free commutative monoids has a $\term{length} : \FF(A) \to \Nat$ function,
where $\Nat$ carries the additive monoid structure $(0,+)$,
which is also a commutative monoid structure since addition is commutative.
\begin{definition}[\alink{definition}{length}{length}]
      \label{def:length}
      The length homomorphism is defined as
      \(
      \ext{(\lambda x.\, 1)} : \FF(A) \to \Nat
      \).
\end{definition}
Further, any presentation of free monoids or free commutative monoids has an
element membership predicate ${\blank\in\blank} : A \to \FF(A) \to \hProp$,
for any set $A$.
Here, we use the fact that $\hProp$ forms a (commutative) monoid under
disjunction and falsehood $(\bot, \vee)$.
\begin{definition}[\alink{definition}{Membership~$\in$}{membership}]
      \label{def:membership}
      The membership predicate on a set $A$ for any element $x:A$ is
      \(
      {x\in\blank} \defeq \ext{\yo_{A}(x)} : {\FF(A) \to \hProp}
      \),
      where we define
      \(
      \yo_A(x) \defeq {\lambda y.\, {x \id y}} : {A \to \hProp}
      \).
\end{definition}
$\yo$ is formally the Yoneda map under the ``types are groupoids'' correspondence,
where $x:A$ is being sent to its representable in the Hom-groupoid (formed by the identity type), of type $\hProp$.
Note that the proofs of (commutative) monoid laws for $\hProp$ use equality,
which requires the use of univalence (or at least, propositional extensionality).
By construction, this membership predicate satisfies its homomorphic properties,
which are colloquially the properties of inductively defined de Bruijn indices.

We note that $\hProp$ is actually one type level higher than $A$.
To make the type levels explicit, $A$ is of type level $\ell$, and since $\hProp_\ell$
is the type of all types $X : \Set_\ell$ that are mere propositions, $\hProp_\ell$ has
type level $\ell + 1$. We do not assume any propositional resizing axioms~\cite{voevodskyResizingRulesTheir2011},
and instead use Agda's universe polymorphism in our formalisation to accommodate this.

Any presentation of free (commutative) monoids $\FF(A)$ also supports the
$\term{Any}$ and $\term{All}$ predicates, which allow lifting a predicate $A \to \hProp$ (on $A$),
to \emph{any} or \emph{all} elements of $\xs : \FF(A)$, respectively.
We note that $\hProp$ forms a (commutative) monoid in two different ways:
$(\bot,\vee)$ and $(\top,\wedge)$ (disjunction and conjunction),
which are the two different ways of getting $\term{Any}$ and $\term{All}$,
respectively, by extension.
\begin{definition}[$\term{Any}$ and $\term{All}$]
      \label{def:any-all}
      \begin{gather*}
            \type{Any}(P) \defeq \ext{P} : \FF(A) \to (\hProp, \bot, \vee)
            \qquad
            \type{All}(P) \defeq \ext{P} : \FF(A) \to (\hProp, \top, \wedge)
      \end{gather*}
\end{definition}
Note that Cubical Agda has problems with indexing over HITs~\cite[\S~8]{ProperSupportInductive,alexandruIntrinsicallyCorrectSorting2023},
hence it is preferable to program with our universal properties, such as when defining $\term{Any}$ and $\term{All}$,
because the indexed-inductive definitions of these predicates get stuck on $\term{transp}$ terms.

There is a $\term{head}$ function on lists, which is a function that returns the first element of a non-empty list.
Formally, this is a monoid homomorphism from $\LL(A)$ to $1 + A$.
\begin{definition}[\alink{definition}{$\term{head}$}{head}]
      \label{def:head-free-monoid}
      The head homomorphism is defined as
      \(
      \term{head} \defeq \ext{\inr} : \LL(A) \to 1 + A
      \),
      where the monoid structure on $1 + A$ has unit
      \(
      e \defeq \inl(\ttt) : 1 + A
      \),
      and multiplication picks the leftmost element that is defined.
      \[
            \begin{array}{rclcl}
                  \inl(\ttt) & \oplus & b & \defeq & b       \\
                  \inr(a)    & \oplus & b & \defeq & \inr(a) \\
            \end{array}
      \]
\end{definition}
This monoid operation $\oplus$ is not commutative,
since swapping the input arguments to $\oplus$ would return the leftmost or rightmost element.
To make it commutative would require a canonical way to pick between a choice of two elements --
this leads us to the next section.

\subsection{Total orders}
\label{sec:total-orders}

First, we recall the axioms of a total order or linear order $\leq$ on a set $A$.
\begin{definition}[\alink{definition}{Total order}{total-order}]
    \label{def:total-order}
    A total order on a set $A$ is a relation $\leq : A \to A \to \hProp$ that satisfies:
    \begin{itemize}
        \item reflexivity: $x \leq x$,
        \item transitivity: if $x \leq y$ and $y \leq z$, then $x \leq z$,
        \item antisymmetry: if $x \leq y$ and $y \leq x$, then $x = y$,
        \item totality: $\forall x, y$, either $x \leq y$ or $y \leq x$.
    \end{itemize}
    A \emph{decidable} total order requires the $\leq$ relation to be decidable:
    \begin{itemize}
        \item decidable totality: $\forall x, y$, we have $x \leq y + \neg(x \leq y)$.
    \end{itemize}
\end{definition}
Note that \emph{either-or} means a (truncated) logical disjunction.
In the context of this paper, we want to make a distinction between
``decidable total order'' and ``total order''.
The decidability axiom strengthens the totality axiom,
where we have either $x \leq y$ or $y \leq x$ merely as a proposition,
but decidability allows us to produce a witness if $x \leq y$ is true.
\begin{proposition}[\alink{proposition}{}{decidable-total-order}]
    \label{prop:decidable-total-order}
    In a decidable total order, it holds that ${\forall x, y}, \ps{x \leq y} + \ps{y \leq x}$.
    Further, this makes $A$ discrete, that is ${\forall x, y}, \ps{x \id y} + \ps{x \neq y}$.
\end{proposition}
\begin{proof}
    We decide if $x \leq y$ and $y \leq x$, and then by cases:
    \begin{itemize}
        \item
              if $x \leq y$ and $y \leq x$: by antisymmetry, $x = y$.
        \item
              if $\neg(x \leq y)$ and $y \leq x$: assuming $x = y$ leads to a contradiction, hence $x \neq y$.
        \item
              if $x \leq y$ and $\neg(y \leq x)$: similar to the previous case.
        \item
              if $\neg(x \leq y)$ and $\neg(y \leq x)$: by totality, either
              $x \leq y$ or $y \leq x$, which leads to a contradiction.
    \end{itemize}
\end{proof}
We also recall the axioms of a \emph{strict} total order $<$ on $A$.
\begin{definition}[\alink{definition}{Strict total order}{strict-total-order}]
    \label{def:strict-total-order}
    A strict total order on a set $A$ is a relation $< : A \to A \to \hProp$ that satisfies:
    \begin{itemize}
        \item irreflexivity: $\neg(x < x)$,
        \item transitivity: if $x < y$ and $y < z$, then $x < z$,
        \item asymmetry: if $x < y$, then $\neg(y < x)$,
        \item cotransitivity: $\forall x, y, z$, if $x < z$, then either $x < y$ or $y < z$.
        \item connectedness: $\forall x, y$, if $\neg(x < y)$ and $\neg(y < x)$, then $x = y$.
    \end{itemize}
    A \emph{decidable} strict total order requires the $<$ relation to be decidable:
    \begin{itemize}
        \item decidability: $\forall x, y$, we have $x < y + \neg(x < y)$.
    \end{itemize}
\end{definition}
\begin{proposition}[\alink{proposition}{}{decidable-strict-total-order}]
    \label{prop:decidable-strict-total-order}
    In a decidable strict total order, it holds that ${\forall x, y}, \ps{x < y} + \ps{y < x}$.
    Further, this makes $A$ discrete, that is ${\forall x, y}, \ps{x \id y} + \ps{x \neq y}$.
\end{proposition}
\begin{proof}
    We decide if $x < y$ and $y < x$, and by cases:
    \begin{itemize}
        \item if $\neg(x < y)$ and $\neg(y < x)$: by connectedness, $x = y$.
        \item if $x < y$ or $y < x$: by irreflexivity, $x \neq y$.
    \end{itemize}
\end{proof}
We can show that giving a decidable \emph{strict} total order structure on $A$
is equivalent to giving a decidable total order structure on $A$.
In practice, however, this equivalence is thoroughly awful to work with in a
proof assistant, leading to horrible goals when formalising proofs that shift
between these representations.
\begin{proposition}[\alink{proposition}{}{decidable-strict-total-order-equiv}]
    \label{prop:decidable-strict-total-order-equiv}
    The set of decidable strict total orders on $A$ and the set of decidable total orders on $A$ are equivalent.
\end{proposition}
\begin{proofsketch}
    Given a decidable total order $\leq$ on $A$, by~\cref{prop:decidable-total-order},
    we map it to a decidable strict total order $<$ on $A$ by
    $x < y \defeq (x \leq y) \times (x \neq y)$.
    Vice versa, by~\cref{prop:decidable-strict-total-order},
    given a decidable strict total order $<$ on $A$,
    we map it to a decidable total order $\leq$ on $A$ by
    $x \leq y \defeq (x < y) + (x = y)$.
    These maps are checked to be inverses of each other.
\end{proofsketch}
An equivalent way of defining a total order is using a binary meet operation, without starting from an ordering relation.
\begin{definition}[\alink{definition}{Meet semi-lattice}{meet-semi-lattice}]
    \label{def:meet-semi-lattice}
    A meet semi-lattice is a set $A$ with a binary operation $\blank\meet\blank : A \to A \to A$ that is:
    \begin{itemize}
        \item idempotent: $x \meet x \id x$,
        \item associative: $(x \meet y) \meet z \id x \meet (y \meet z)$,
        \item commutative: $x \meet y \id y \meet x$.
    \end{itemize}
    A \emph{total} meet semi-lattice further satisfies:
    \begin{itemize}
        \item total: $\forall x, y$, either $x \meet y \id x$ or $x \meet y \id y$.
    \end{itemize}
    A \emph{decidable} total meet semi-lattice strengthens this to:
    \begin{itemize}
        \item decidable totality: ${\forall x, y}, \ps{x \meet y \id x} + \ps{x \meet y \id y}$.
    \end{itemize}
\end{definition}

\begin{proposition}[\alink{proposition}{}{discrete-meet-semi-lattice}]
    \label{prop:discrete-meet-semi-lattice}
    If $A$ has decidable equality, all total meet semi-lattice structures on $A$ are decidable.
\end{proposition}
\begin{proof}
    We decide if $x \meet y \id x$ and $x \meet y \id y$ to compute whether
    $x \meet y \id x$ or $x \meet y \id y$. By totality, one of them must hold.
\end{proof}

\begin{proposition}[\alink{proposition}{}{total-order-meet-semi-lattice}]
    \label{prop:total-order-meet-semi-lattice}
    A total order structure on a set $A$ is equivalent to a total meet semi-lattice structure on $A$.
    Further, a decidable total order on $A$ induces a decidable total meet semi-lattice structure on $A$.
\end{proposition}
\begin{proofsketch}
    Given a (mere) total order $\leq$ on a set $A$,
    we define ${x \meet y} \defeq \term{if} x \leq y \term{then} x \term{else} y$.
    Crucially, this map is \emph{locally-constant}, allowing us to eliminate from an $\hProp$ to an $\hSet$.
    Meets satisfy the universal property of products, that is,
    ${c \leq a \meet b} \Leftrightarrow {c \leq a} \land {c \leq b}$,
    and the axioms follow by calculation using $\yo$-arguments.
    Conversely, given a meet semi-lattice, we define $x \leq y \defeq x \meet y \id x$,
    which defines an $\hProp$-valued total ordering relation.
    If the total order is decidable, by~\cref{prop:decidable-total-order} and~\cref{prop:discrete-meet-semi-lattice},
    the mapped meet operation is decidable.
    The converse does not hold however,
    as a decidable total meet semi-lattice structure does not imply discreteness of $A$.
\end{proofsketch}
Finally, tying this back to~\cref{def:head-free-monoid}, we have the following for free commutative monoids.
\begin{definition}[\alink{definition}{$\term{head}$}{head-free-commutative-monoid}]
    \label{def:head-free-commutative-monoid}
    Assume a total order $\leq$ on a set $A$.
    We define a commutative monoid structure on $1 + A$,
    with unit \(e \defeq \inl(\ttt) : 1 + A\), and multiplication defined as:
    \[
        \begin{array}{rclcl}
            \inl(\ttt) & \oplus & b          & \defeq & b                         \\
            \inr(a)    & \oplus & \inl(\ttt) & \defeq & \inr(a)                   \\
            \inr(a)    & \oplus & \inr(b)    & \defeq & \inr(a \meet b) \enspace.
        \end{array}
    \]
    This gives a homomorphism \({\term{head} \defeq \ext{\inr}} : {\MM(A) \to 1 + A}\),
    which picks out the \emph{least} element of the free commutative monoid.
\end{definition}

\subsection{Sorting functions}
\label{sec:sorting}

The free commutative monoid is also a monoid, hence, there is a canonical monoid homomorphism
$\quotient : \LL(A) \to \MM(A)$, which is given by $\ext{\eta_A}$, the extension of the unit $\eta_A : A \to \MM(A)$.
Since $\MM(A)$ is (upto equivalence), a quotient of $\LL(A)$ by symmetries (or a permutation relation),
it is a surjection (in particular, a regular epimorphism in $\Set$ as constructed in type theory),
the canonical inclusion into the quotient.
Concretely, $\quotient$ simply includes the elements of $\LL(A)$ into equivalence classes of lists in $\MM(A)$,
which ``forgets'' the order that was imposed by the indexing of the list.

Classically, assuming the Axiom of Choice would allow us to construct a section ({right-inverse}, in $\Set$) to the surjection $\quotient$,
that is,
a function $\ssection : \MM(A) \to \LL(A)$ such that $\forall x.\, \quotient(\ssection(x)) \id x$.
Or in informal terms, given the surjective inclusion into the quotient,
a section (uniformly) picks out a canonical representative for each equivalence class.
The core question we want to study is the existence of $\ssection$ in a constructive
setting, or equivalently, whether the order factored out by the symmetry quotient
can be constructively recovered.
\begin{figure}
    \centering
    \scalebox{1.0}{
        \begin{tikzcd}[ampersand replacement=\&,cramped]
            {\LL(A)} \&\&\& {\MM(A)}
            \arrow["\ssection", curve={height=-10pt}, from=1-4, to=1-1]
            \arrow["\quotient", two heads, from=1-1, to=1-4]
        \end{tikzcd}
    }
    \caption{Relationship of $\LL(A)$ and $\MM(A)$}
    \label{fig:enter-label}
\end{figure}
Viewing the quotienting relation as a permutation relation (from~\cref{sec:cmon:qfreemon}), a section $\ssection$ to $\quotient$ has to pick out
canonical representatives of equivalence classes generated by permutations.
Using $\SList$ as an example,
$\ssection(x \cons y \cons \xs) \id \ssection(y \cons x \cons \xs)$ for any $x, y : A$ and $\xs : \SList(A)$,
by $\term{swap}$.
Since $\forall \xs.\,\quotient(\ssection(\xs)) \id \xs$, $\ssection$ must preserve all the elements of $\xs$.
It cannot be a trivial function such as $\lambda\,\xs. []$ -- it must produce a permutation of the elements of $\xs$.
But to place these elements side-by-side in the list, $\ssection$ must somehow impose an order on $A$
(invariant under permutation), turning unordered lists of $A$ into ordered lists of $A$.
Axiom of Choice (AC) giving us a section $\ssection$ to $\quotient$ ``for free'' is analagous to how
AC implies the well-ordering principle, which states every set can be well-ordered.
Thus, if AC was assumed, we could easily recover an order on $A$ from the section $\ssection$.
Instead we study how to constructively define such a section, and in fact,
that is exactly the extensional view of a sorting algorithm,
and the implications of its existence is that $A$ can be ordered, or sorted.

\subsubsection{Section from Order}
\label{sec:sort-section-from-order}

\begin{proposition}[\alink{proposition}{}{sort-from-order}]
    Assume a decidable total order on $A$. There is a sort function $\ssection: \MM(A) \to \LL(A)$
    which constructs a section to $\quotient : \LL(A) \twoheadrightarrow \MM(A)$
\end{proposition}

\begin{proofsketch}
    We may construct such a sort function by implementing any sorting algorithm.
    In our formalisiation we chose insertion sort,
    because it can be defined easily using the inductive structure of $\SList(A)$ and $\List(A)$.
    To implement other sorting algorithms like mergesort,
    other representations such as $\Bag$ and $\Array$ would be preferable,
    as explained in~\cref{sec:cmon:bag}.
    To see how this proposition holds: $\quotient(\ssection(\xs))$ first orders an unordered list $\xs$ by $\ssection$,
    then forgets the order again by $\quotient$ --
    imposing and then forgetting an order on $\xs$ simply \emph{permutes} its elements,
    which proves $\quotient \comp \ssection = \idfunc$.
\end{proofsketch}

\subsubsection{Order from Section}
\label{sec:sort-order-from-section}

The previous section allowed us to construct a section, but an arbitrary section may not be a sorting function.
To show a section is indeed a sort function, we need to show the section imposes some total order on $A$ which it sorts by.
Indeed, just by the virtue of $\ssection$ being a section,
we can \emph{almost} construct a total-order on the carrier set.

\begin{definition}[\alink{definition}{$\term{least}$}{least}]
    \label{def:least}
    Given a section $\ssection$, we define a relation $\leqs$ parametrised by $\ssection$:
    \[
        x \leqs y \defeq \term{head}(\ssection(\bag{x, y})) = \inr(x) \enspace.
    \]
\end{definition}
Recall that the head homomorphism $\term{head} \defeq \ext{\inr} : \LL(A) \to 1 + A$ (as in~\cref{def:head-free-monoid}) is defined by equipping $1 + A$ with a monoid structure with unit $e \defeq \inl(\ttt)$ and multiplication defined as $\inl(\ttt) \oplus b \defeq b$ and $\inr(a) \oplus b \defeq \inr(a)$, which picks out the first (leftmost) element of a list.
Here, we take the two-element bag $\bag{x, y}$,
``sort'' it by $\ssection$, and test if the $\term{head}$ element is $x$.
Note, this is equivalent to $x \leqs y \defeq s\bag{x, y} = [x,y]$,
because $\ssection$ preserves length, and the second element is forced to be $y$.
%

\begin{proposition}[\alink{proposition}{}{sort-almost-order}]
    \label{sort:almost-order}
    $\leqs$ is reflexive, antisymmetric, and total.
\end{proposition}
\begin{proof}
    For all $x$, $\term{least}(\bag{x, x})$ must be $\inr(x)$, therefore $x \leqs x$, giving reflexivity.
    For all $x$ and $y$, given $x \leqs y$ and $y \leqs x$,
    we have $\term{least}(\bag{x, y}) = \inr(x)$ and $\term{least}(\bag{y, x}) = \inr(y)$.
    Since $\bag{x, y} = \bag{y, x}$, $\term{least}(\bag{x, y}) = \term{least}(\bag{y, x})$,
    therefore we have $x = y$, giving antisymmetry.
    For all $x$ and $y$, $\term{least}(\bag{x, y})$ is merely either $\inr(x)$ or $\inr(y)$,
    therefore we have merely either $x \leqs y$ or $y \leqs x$, giving totality.
\end{proof}
A crucial observation is that $\ssection$ correctly orders 2-element bags,
but it does not necessarily sort bags with 3 or more elements.
\begin{proposition}[\alink{proposition}{}{counterexample-transitivity}]
    \label{prop:counterexample-transitivity}
    $\leqs$ is not necessarily transitive.
\end{proposition}
\begin{proof}
    A sorting section $\ssection$ would indeed give a $\leqs$
    that satisfies transitivity. However, we can turn a sorting section
    into a non-sorting misbehaving section that gives a non-transitive $\leqs$.
    Take the sorting section
    $\term{sort} : \SList(\Nat) \to \List(\Nat)$ which sorts $\SList(\Nat)$ ascendingly,
    we can use $\term{sort}$ to construct a counterexample $\ssection$ as follows:
    \begin{align*}
        \ssection(\xs) & \defeq
        \begin{cases}
            [3, 1]           & \text{if $\xs = \bag{1, 3}$} \\
            \term{sort}(\xs) & \text{otherwise}             \\
        \end{cases}
    \end{align*}
    Indeed, we have $1 \leqs 2$ and $2 \leqs 3$ as expected,
    but $\term{least}(\bag{1, 3}) = \inr(3)$, therefore $1 \not\leqs 3$,
    thus violating transitivity.
\end{proof}
To make sure that $\ssection$ is a well-behaved section that recovers a total order on $A$,
we will enforce additional constraints on the \emph{image} of $\ssection$.
\begin{definition}[\alink{definition}{$\blank\in\im{\ssection}$}{in-image}]
    \label{def:in-image}
    The fiber of $\ssection$ at a point in the codomain~$\xs : \LL(A)$
    is given by $\fib_{\ssection}(\xs) \defeq \dsum{ys : \MM(A)}{(\ssection(ys) = \xs)}$.
    The image of $\ssection$ is given by $\im{\ssection} \defeq \dsum{\xs : \LL(A)}{\Trunc[-1]{\fib_{\ssection}(\xs)}}$.
    Simplifying, we say that $\xs:\LL(A)$ is ``in the image of $\ssection$'', or, $\xs \in \im{\ssection}$,
    if there merely exists a $\ys:\MM(A)$ such that $\ssection(\ys) = \xs$.
\end{definition}
If $\ssection$ \emph{were} a sort function, the image of $\ssection$ would be the set of $\ssection$-``sorted'' lists,
therefore $\inimage{\xs}$ would imply $\xs$ is a correctly $\ssection$-``sorted'' list.
First, we note that the 2-element case is correct.
\begin{proposition}[\alink{proposition}{}{sort-to-order}]
    \label{sort:sort-to-order}
    $x \leqs y$ \; iff \; $\inimage{[x, y]}$.
\end{proposition}
\noindent Then, we state the first axiom on $\ssection$.
\begin{definition}[\alink{definition}{$\isheadleast$}{head-least}]
    \label{sort:head-least}
    A section $\ssection$ satisfies $\isheadleast$ iff for all $x, y, \xs$:
    \[
        y \in x \cons \xs \; \land \; \inimage{x \cons \xs} \; \to \; \inimage{[x, y]}
        \enspace.
    \]
\end{definition}
\noindent
There are two different membership symbols in this axiom,
where the first membership is list membership from~\cref{def:membership},
and the second membership is image membership from~\cref{def:in-image}.
The $\in$ symbol is intentionally overloaded to make the axiom look like a logical ``cut'' rule,
that pushes the (least) $x$ element to the head of the list.
Informally, the head of an $\ssection$-``sorted'' list is always the least element of the list.

\begin{proposition}[\alink{proposition}{}{order-to-sort-head-least}]
    \label{prop:order-to-sort-head-least}
    If $A$ has a total order $\leq$,
    the insertion sort function defined using $\leq$ satisfies $\isheadleast$.
\end{proposition}

\begin{proposition}[\alink{proposition}{}{trans}]
    \label{sort:trans}
    If $\ssection$ satisfies $\isheadleast$, $\leqs$ is transitive.
\end{proposition}
\begin{proof}
    Given $x \leqs y$ and $y \leqs z$, we want to show $x \leqs z$.
    Consider the 3-element bag $\bag{x,y,z} : \MM(A)$.
    Let $u$ be $\term{least}(\bag{x,y,z})$,
    by~\cref{sort:head-least} and~\cref{sort:sort-to-order},
    we have $u \leqs x \land u \leqs y \land u \leqs z$.
    Since $u \in \bag{x,y,z}$, $u$ must be one of the elements.
    If $u = x$ we have $x \leqs z$.
    If $u = y$ we have $y \leqs x$,
    and since $x \leqs y$ and $y \leqs z$ by assumption,
    we have $x = y$ by antisymmetry, and then we have $x \leqs z$ by substitution.
    Finally, if $u = z$, we have $z \leqs y$, and since $y \leqs z$ and $x \leqs y$ by assumption,
    we have $z = y$ by antisymmetry, and then we have $x \leqs z$ by substitution.
\end{proof}

\subsubsection{Embedding orders into sections}
\label{sec:sort-embedding}

Following from \cref{sort:almost-order,sort:trans},
and \cref{prop:order-to-sort-head-least},
we have shown that a section $\ssection$ that satisfies $\isheadleast$ produces a total order
$x \leqs y \defeq \term{least}(\bag{x, y}) \id \inr(x)$,
and a total order $\leq$ on the carrier set produces a section satisfying $\isheadleast$,
constructed by sorting with $\leq$.
This constitutes an embedding of decidable total orders into sections satisfying $\isheadleast$.

\begin{proposition}[\alink{proposition}{}{o2s2o}]
    \label{sort:o2s2o}
    Assume $A$ has a decidable total order $\leq$, we can construct a section $\ssection$ that
    satisfies $\isheadleast$, such that $\leqs$ constructed from $\ssection$ is equivalent
    to $\leq$.
\end{proposition}
\begin{proof}
    By the insertion sort algorithm parameterized by $\leq$,
    it holds that $\inimage{[x, y]}$ iff $x \leq y$.
    By~\cref{sort:sort-to-order}, we have $x \leqs y$ iff $x \leq y$.
    We now have a total order $x \leqs y$ that is equivalent to $x \leq y$.
\end{proof}

\subsubsection{Equivalence of order and sections}

We want to improve the embedding to an isomorphism, and it
remains to show that we can turn a section satisfying $\isheadleast$ to a total order $\leqs$,
then reconstruct the \emph{same} section from $\leqs$.
Unfortunately, $\isheadleast$ is not enough to guarantee this.
\begin{proposition}[\alink{proposition}{}{counterexample-equivalence}]
    \label{prop:counterexample-equivalence}
    There is no equivalence between sections satisfying the $\isheadleast$ axiom and total orders on $A$.
\end{proposition}
\begin{proof}
    We consider the set $\Nat$ of natural numbers,
    and construct a $\ssection$ that satisfies $\isheadleast$ but is not a sort function.
    Let the correct insertion sort function be $\term{sort} : \MM(\Nat) \to \LL(\Nat)$.
    We define $\term{reverseTail}$ which reverses only the tail of a list,
    and a section $\ssection$ that sorts correctly using $\term{sort}$ then reverses the tail:
    \begin{align*}
        \term{reverseTail}([])          & \defeq []                                   \\
        \term{reverseTail}(x \cons \xs) & \defeq x \cons \term{reverse}(\xs)          \\
        \ssection(\xs)                  & \defeq \term{reverseTail}(\term{sort}(\xs))
    \end{align*}
    For example:
    \begin{align*}
        \ssection(\bag{2,3,1,4}) & = [1,4,3,2] \\
        \ssection(\bag{2,3,1})   & = [1, 3, 2] \\
        \ssection(\bag{2,3})     & = [2, 3]
    \end{align*}
    Note that both $\term{sort}$ and $\ssection$ satisfy $\isheadleast$,
    but $\ssection$ only sorts 2-element lists correctly, and $\ssection \neq \term{sort}$.
    By~\cref{sort:o2s2o} we can use both $\term{sort}$ and $\ssection$
    to reconstruct the same $\leq$ on $\Nat$.
    Hence, two different sections satisfying $\isheadleast$ map to the same total order.
\end{proof}
Therefore, we need to introduce a second axiom of sorting.

\begin{definition}[\alink{definition}{$\istailsort$}{tail-sort}]
    \label{def:tail-sort}
    A section $\ssection$ satisfies $\istailsort$ iff
    for all $x, \xs$,
    \[
        \inimage{x \cons \xs} \to \inimage{\xs}
    \]
\end{definition}
This says that $\ssection$-``sorted'' lists are downwards-closed under cons-ing, that is,
the tail of an $\ssection$-``sorted'' list is also $\ssection$-``sorted''.
To prove the correctness of our axioms,
first we need to show that a section $\ssection$ satisfying
$\isheadleast$ and $\istailsort$ is equal to insertion sort parameterized by
the $\leqs$ constructed from $\ssection$.
In fact, the axioms we have introduced are equivalent to the standard inductive characterisation of sorted lists,
found in textbooks, such as in~\cite{appelVerifiedFunctionalAlgorithms2023}.

\begin{definition}[\alink{definition}{$\isSorted$}{is-sorted}]
    \label{def:is-sorted}
    Given a total order $\leq$ on $A$,
    the predicate $\isSorted_{\leq}$ on $\LL(A)$ is generated by the following constructors:
    \begin{align*}
        \term{sorted-nil}  & : \isSorted_{\leq}([])              \\
        \term{sorted-\eta} & : \forall x.\,\isSorted_{\leq}([x]) \\
        \term{sorted-cons} & : \forall x, y, \zs.\, x \leq y
        \to \isSorted_{\leq}(y \cons \zs)
        \to \isSorted_{\leq}(x \cons y \cons \zs)
    \end{align*}
\end{definition}
Note that $\isSorted_{\leq}(\xs)$ is a proposition for every $\xs$,
and forces the list $\xs$ to be permuted in a unique way.
This predicate is equivalent to the standard $\type{Sorted}$ relation in the Agda standard library~\cite{Daggitt_The_Agda_standard_2025}, which is defined pairwise using the more general $\type{Linked}$ relation (i.e., $\type{Sorted}(\xs) \defeq \type{Linked}(\blank\leq\blank, \xs)$).
Indeed, many of these definitions are traditional constructions that are completely independent of univalence or cubical features.
\begin{lemma}[\alink{lemma}{}{is-sorted-unique}]
    Given a total order $\leq$, for any lists $\xs, \ys : \LL(A)$,
    if $\quotient(\xs) = \quotient(\ys)$, $\isSorted_{\leq}(\xs)$, and $\isSorted_{\leq}(\ys)$,
    then $\xs = \ys$.
\end{lemma}
Insertion sort by $\leq$ always produces lists that satisfy $\isSorted_{\leq}$.
Sections that also produce lists satisfying $\isSorted_{\leq}$ are equal to insertion sort
by function extensionality.
\begin{proposition}[\alink{proposition}{}{sort-uniq}]
    \label{sort:sort-uniq}
    Given a total order $\leq$,
    if a section $\ssection$ always produces a sorted list, that is, $\forall \xs.\,\isSorted_{\leq}(\ssection(\xs))$,
    $\ssection$ is equal to insertion sort by $\leq$.
\end{proposition}
Finally, this gives us the correctness property of our axioms.
\begin{proposition}[\alink{proposition}{}{well-behave-sorts}]
    \label{sort:well-behave-sorts}
    Given a section $\ssection$ that satisfies $\isheadleast$ and $\istailsort$,
    and $\leqs$ the order derived from $\ssection$, then for all $\xs : \MM(A)$,
    it holds that $\isSorted_{\leqs}(\ssection(\xs))$.
    Equivalently, for all lists $\xs : \LL(A)$,
    it holds that
    $\xs \in \im{s}$ iff $\isSorted_{\leqs}(\xs)$.
\end{proposition}
\begin{proof}
    We inspect the length of $\xs : \MM(A)$.
    For lengths 0 and 1, this holds trivially.
    Otherwise, we proceed by induction:
    given a $\xs : \MM(A)$ of length $\geq 2$, let $\ssection(\xs) = x \cons y \cons \ys$.
    We need to show
    $x \leqs y \land \isSorted_{\leqs}(y \cons \ys)$ to construct
    $\isSorted_{\leqs}(x \cons y \cons \ys)$.
    By $\isheadleast$, we have $x \leqs y$, and by $\istailsort$, we
    inductively prove $\isSorted_{\leqs}(y \cons \ys)$.
\end{proof}

\begin{lemma}[\alink{lemma}{}{s2o2s}]
    \label{sort:s2o2s}
    Given a decidable total order $\leq$ on $A$, we can construct
    a section $t_\leq$ satisfying the axioms $\isheadleast$ and $\istailsort$,
    such that, for the order $\leqs$ derived from~$\ssection$,
    we have $t_{\leq} = \ssection$.
\end{lemma}
\begin{proof}
    From $\ssection$ we can construct a decidable total order $\leqs$ since $\ssection$ satisfies
    $\isheadleast$ and $A$ has decidable equality by assumption.
    We construct $t_{\leqs}$ as insertion sort
    parameterized by $\leqs$ constructed from $\ssection$.
    By ~\cref{sort:sort-uniq} and ~\cref{sort:well-behave-sorts}, $s = t_{\leqs}$.
\end{proof}

\noindent
We can now state and prove our main theorem.
\begin{definition}[\alink{definition}{Sorting function}{sorting-function}]
    \label{def:sorting-function}
    \leavevmode
    A sorting function is a section $\ssection : \MM(A) \to \LL(A)$ to
    the canonical surjection $\quotient : \LL(A) \twoheadrightarrow \MM(A)$ satisfying two axioms:
    \begin{itemize}[leftmargin=*]
        \item $\isheadleast :
                  \forall\, x, y, \xs,\,\,
                  y \in x \cons \xs \mathrel{\land} \inimage{x \cons \xs} \to \inimage{[x, y]}
              $,
        \item $\istailsort :
                  \forall\, x, \xs,\,\,
                  \inimage{x \cons \xs} \to \inimage{\xs}
              $.
    \end{itemize}
\end{definition}
\begin{theorem}[\alink{theorem}{}{main}]
    \label{sort:main}
    Let $\term{DecTotOrd}(A)$ be the set of decidable total orders on $A$,
    $\term{Sort}(A)$ be the set of sorting functions with carrier set $A$,
    and $\term{Discrete}(A)$ be a predicate which states $A$ has decidable equality.
    There is a map $o2s \colon \term{DecTotOrd}(A) \to \term{Sort}(A) \times \term{Discrete}(A)$,
    which is an equivalence.
\end{theorem}
\begin{proof}
    $o2s$ is constructed by parameterizing insertion sort with $\leq$.
    By~\cref{prop:decidable-total-order}, $A$ is $\term{Discrete}$.
    The inverse $s2o(s)$ is constructed by~\cref{def:least}, which produces
    a total order by~\cref{sort:almost-order,sort:trans},
    and a decidable total order by $\term{Discrete}(A)$.
    By~\cref{sort:o2s2o} we have $s2o \comp o2s \id \idfunc$,
    and by~\cref{sort:s2o2s} we have $o2s \comp s2o \id \idfunc$,
    giving an isomorphism, hence an equivalence.
\end{proof}

\begin{corollary}[\alink{corollary}{}{strict-order}]
    \label{sort:strict-order}
    Let $\term{DecSTotOrd}(A)$ be the set of decidable strict total orders on $A$,
    There is a map $t2s \colon \term{DecSTotOrd}(A) \to \term{Sort}(A) \times \term{Discrete}(A)$,
    which is an equivalence.
\end{corollary}
\begin{proof}
    By~\cref{prop:decidable-total-order}, we have an equivalence
    $\term{DecSTotOrd}(A) \simeq \term{DecTotOrd}(A)$, and the result follows from~\cref{sort:main}
    by transitivity of equivalences.
\end{proof}

\begin{corollary}[\alink{corollary}{}{lattice}]
    \label{sort:lattice}
    Assume $A$ has decidable equality.
    Let $\term{TMSLat}(A)$ be the set of total meet semilattices on $A$,
    There is a map $l2s \colon \term{TMSLat}(A) \to \term{Sort}(A)$,
    which is an equivalence.
\end{corollary}
\begin{proof}
    By assumption and~\cref{prop:discrete-meet-semi-lattice}, every total meet semilattice on $A$ is decidable.
    Thus, by~\cref{prop:total-order-meet-semi-lattice}, we have an equivalence
    $\term{TMSLat}(A) \simeq \term{DecTotOrd}(A)$, and the result follows from~\cref{sort:main}
    by transitivity of equivalences.
\end{proof}
The sorting axioms we have come up with are abstract properties of functions.
As a sanity check, we can verify that the colloquial correctness specification of a sorting function (starting from a
total order) matches our axioms. We use the correctness criterion developed in~\cite{alexandruIntrinsicallyCorrectSorting2023}.
\begin{proposition}[\alink{proposition}{}{sort-correctness}]
    \label{prop:sort-correctness}
    Assume a decidable total order $\leq$ on $A$.
    A sorting algorithm is a map $\term{sort} : {\LL(A) \to \OLL(A)}$,
    that turns lists into ordered lists,
    where $\OLL(A)$ is defined as $\dsum{\xs : \LL(A)}{\isSorted_{\leq}(\xs)}$,
    such that:
    \[\begin{tikzcd}
            {\LL(A)} && {\OLL(A)} \\
            & {\MM(A)}
            \arrow["{\term{sort}}", from=1-1, to=1-3]
            \arrow["q"', from=1-1, to=2-2]
            \arrow["{q \comp \pi_1}", from=1-3, to=2-2]
        \end{tikzcd}\]
    Sorting functions (as in~\cref{def:sorting-function}) give sorting algorithms.
\end{proposition}
\begin{proof}
    We construct our section $\ssection:\MM(A) \to \LL(A)$,
    and define $\term{sort} \defeq \ssection \comp \quotient$,
    which produces ordered lists by~\cref{sort:well-behave-sorts}.
\end{proof}

\section{Discussion}
\label{sec:discussion}

We conclude by discussing some high-level observations, related work, and future directions.

\myparagraph{Formalisation}
The paper uses informal type theoretic language,
and is accessible without understanding any details of the formalisation.
However, the formalisation is done in Cubical Agda, which has a few differences and
shortcomings, that cause certain details to be more verbose than they are in informal type theory.

For simplicity we omitted type levels in the paper,
but our formalisation has many verbose uses of universe levels due to Agda's universe polymorphism.
Similarly, h-levels were restricted to sets in the paper, but the formalisation is parameterised in many places for any
h-level (to facilitate future generalisations).
The free algebra framework currently only works with sets.
Due to issues of regularity, certain computations only hold propositionally,
and the formalisation requires proving auxiliary $\beta$ and $\eta$ computation rules in some places.

\myparagraph{Universal Algebra}
One of the contributions of our work is a rudimentary framework for universal algebra in a more
categorical style, which lends itself to an elegant formalisation in type theory.
This framework may be improved and generalised from sets to groupoids,
using a system of coherences on top of a system of equations,
to talk about 2-algebraic structures on groupoids,
such as monoidal groupoids, symmetric monoidal groupoids, and so on.

\myparagraph{Free commutative monoids}
The construction of finite multisets and free commutative monoids in type theory has a long history,
and various authors have different approaches to it. We refer the reader to the discussions
in~\cite{choudhuryFreeCommutativeMonoids2023,joramConstructiveFinalSemantics2023} for a detailed survey of these
constructions.
Our work, in particular, was motivated by the colloquial observation that:
``there is no way to represent free commutative monoids using inductive types''.
From the categorical point of view, note that the free monoid functor on $\Set$ is polynomial,
but the free commutative monoid endofunctor is not polynomial, since it only weakly preserves pullbacks.
Various authors have given clever encodings of free commutative monoids using inductive types by adding
assumptions on the carrier set -- in particular, the assumption of total ordering on the carrier set leads to the
construction of ``fresh-lists'', by~\cite{kupkeFreshLookCommutativity2023}, which forces the canonical \emph{sorted}
ordering on the elements of the finite multiset.

It is worth noting that in programming practice,
it is usually the case that all user-defined types have some sort of total order enforced on them,
either because they're finite, or because they can be canonically enumerated.
Under these assumptions, the construction of fresh lists is a very reasonable way to represent free
commutative monoids, or finite multisets.

The free monoid and free commutative monoid constructions can be categorified
to free monoidal categories and free symmetric monoidal categories, respectively.
In type theory, these types can be groupoidified to free monoidal groupoids and free symmetric monoidal groupoids,
using the groupoid structure of identity types.
This is a natural next step,
and its connections to assumptions about total orders on the type of objects would be an important direction to explore.

\myparagraph{Correctness of Sorting}
Sorting is a classic problem in computer science --
the programming point of view of sorting and its correctness has been studied by various authors.
The simplest view of sorting is a function $\term{sort}: \LL(\Nat) \to \LL(\Nat)$,
which permutes the list and outputs an ordered list, which is studied in~\cite{appelVerifiedFunctionalAlgorithms2023}.
Fundamentally, this is a very extrinsic view of program verification,
which is commonly used in program verification and proof assistants,
and further, a special case of a more general sorting algorithm.

Henglein~\cite{hengleinWhatSortingFunction2009} studies sorting functions abstractly,
without requiring a total order on the underlying set.
He considers sorting functions as functions on sequences (lists),
and recovers the order by applying a ``sorting function'' on an $n$-element list,
and looking up the position of the elements to be compared.
Unlike our approach, he does not factorize sorting functions through free commutative monoids.
Henglein also studies sorting on preorders where antisymmetry may not hold,
whereas we have only considered total orders in this paper.
We are able to give a more refined axiomatisation of sorting because we consider permutations
explicitly, and work in a constructive setting (using explicit assumptions about decidability),
and which is an improvement over previous work.

Another intrinsic view of correctness of sorting has been studied in~\cite{hinzeSortingBialgebrasDistributive2012}
using bialgebras, which was further expanded in~\cite{alexandruIntrinsicallyCorrectSorting2023} with an accompanying formalisation in Cubical Agda,
that matches our point of view, as explained in~\cref{prop:sort-correctness}.
However, their work is not just about extensional correctness of sorting,
but also deriving various sorting algorithms starting from bialgebraic semantics and distributive laws.
Our work is complementary, in that we are not concerned with the computational content of sorting, but rather
the abstract properties of sorting functions, which are independent of a given ordering.
It is unclear whether the abstract property of sorting functions can be combined
with the intrinsic complexity of sorting algorithms --
and that is a direction for future work.

Although we only talk about sorting lists and bags,
the same ideas can be applied to other inductive types (polynomial functors) such as trees.
We speculate that this could lead to some interesting connections with sorting
(binary) trees, and constructions of (binary) search trees, from classical computer science.


\renewcommand{\appendixsectionformat}[2]{
  {Supplementary material for Section~#1 (#2)}
}

\bibliography{symmetries}
\appendix

\end{document}
